\documentclass[letterpaper,11pt]{article}
\pdfoutput=1
\usepackage[letterpaper,margin=1in]{geometry}

\usepackage[utf8]{inputenc}
\usepackage{amsmath,amsthm,amssymb}
\usepackage[margin=1in]{geometry}
\usepackage{xcolor,xspace}
\usepackage{cleveref}
\usepackage{paralist}

\newtheorem{theorem}{Theorem}[section]
\newtheorem{definition}[theorem]{Definition}
\newtheorem{lemma}[theorem]{Lemma}

\newtheorem*{claim*}{Claim}
\newtheorem{remark}[theorem]{Remark}

\newtheorem{proposition}[theorem]{Proposition}

\newcommand{\LOCAL}{LOCAL\xspace}
\newcommand{\CONGEST}{CONGEST\xspace}

\DeclareMathOperator*{\E}{\mathbb{E}}
\DeclareMathOperator{\polylog}{polylog}
\DeclareMathOperator{\polyloglog}{polyloglog}
\DeclareMathOperator{\poly}{poly}

\DeclareMathOperator{\ID}{ID}
\newcommand{\alg}[1]{\textsc{#1}}
\newcommand{\mypar}[1]{\paragraph{#1}}

\newcommand{\set}[1]{\left\{#1\right\}}

\newcommand{\Otilde}{{\widetilde O}}
\newcommand{\Omegatilde}{{\widetilde \Omega}}

\newcommand{\ND}{\mathit{ND}_2(\log n)}

\newcommand{\calA}{{\mathcal A}}

\newcommand{\Deltahat}{\hat{\Delta}}
\newcommand{\Vfr}{V^{\text{friendly}}}  
\newcommand{\custombox}[1]{#1}

\newcommand{\myemail}[1]{\,$\cdot$\, {\small #1}}
\newcommand{\myaff}[1]{\,$\cdot$\, {\small #1}\par\smallskip}

\newenvironment{mycover}
{\list{}{\listparindent 0pt
        \itemindent    \listparindent
        \leftmargin    1cm
        \rightmargin   0.5cm
        \parsep        0pt}%
    \raggedright
    \item\relax}
{\endlist}

\begin{document}
\begin{mycover}
{\huge\bfseries\boldmath 
Coloring Fast Without Learning Your Neighbors' Colors 
\par}
\bigskip
\bigskip
\bigskip

\textbf{Magn\'us M. Halld\'orsson}
\myemail{mmh@ru.is}
\myaff{Reykjavik University, Iceland}

\textbf{Fabian Kuhn}
\myemail{kuhn@cs.uni-freiburg.de}
\myaff{University of Freiburg, Germany}

\textbf{Yannic Maus}
\myemail{yannic.maus@cs.technion.ac.il}
\myaff{Technion, Israel}

\textbf{Alexandre Nolin}
\myemail{alexandren@ru.is}
\myaff{Reykjavik University, Iceland}

\end{mycover}
\bigskip

\begin{abstract} 
We give an improved randomized CONGEST algorithm for distance-$2$ coloring that uses $\Delta^2+1$ colors and runs in $O(\log n)$ rounds, improving the recent $O(\log \Delta \cdot \log n)$-round algorithm in [Halld\'orsson, Kuhn, Maus; PODC '20]. We then improve the time complexity to $O(\log \Delta) + 2^{O(\sqrt{\log\log n})}$. 
\end{abstract}
\clearpage
\tableofcontents
\clearpage
\section{Introduction} 
\label{sec:intro}
The distributed coloring problem is arguably the most intensively studied problem in the area of distributed graph algorithms and certainly also one of the most intensively studied problems in distributed computing more generally. The standard assumption is that the \emph{coloring graph} -- the graph on which we want to compute a coloring -- is also the \emph{communication network} -- the graph forming the network topology. We explore in this paper the case when the latter is weaker than the former: the communication is constrained, and direct links are not available to all the ``neighbors'' that are to be colored differently.

The primary setting for this is the \emph{distance-$2$ coloring} problem in the standard distributed \CONGEST model. Given a graph $G=(V,E)$, in the \emph{d2-coloring} problem on $G$, the objective is to assign a color $x_v$ to each node $v\in V$ such that any two nodes $u$ and $v$ at distance at most $2$ in $G$ are assigned different colors $x_u\neq x_v$. Equivalently, d2-coloring asks for a coloring of the nodes of $G$ such that for every $u\in V$, all the nodes in the set $\set{u}\cup N(u)$ (where $N(u)$ denotes the set of neighbors of $u$) are assigned distinct colors. Further note that d2-coloring on $G$ is also equivalent to the usual vertex coloring problem on the graph $G^2$, where $V(G^2)=V$ and there is an edge $\set{u,v}\in E(G^2)$ whenever $d_G(u,v)\leq 2$.

The \CONGEST model is a standard synchronous message passing model~\cite{peleg00}. The graph on which we want to compute a coloring is also assumed to form the network topology. Each node $u\in V$ of the graph has a unique $O(\log n)$-bit identifier $\ID(u)$, where $n=|V|$ is the number of nodes of $G$. Time is divided into synchronous rounds and in each round, each node $u\in V$ of $G$ can do some arbitrary internal computation, send a (potentially different) message to each of its neighbors $v\in N(u)$, and receive the messages sent by its neighbors in the current round. If the size of the messages is not restricted, the model is known as the \LOCAL model~\cite{linial92,peleg00}. In the \CONGEST model, it is further assumed that each message consists of at most $O(\log n)$ bits. 

As our main result, we give an efficient $O(\log n)$-time randomized algorithm for d2-coloring $G$ with at most $\Delta^2+1$ colors, where $\Delta$ is the maximum degree of $G$. This improves on a recent $O(\log \Delta \cdot \log n)$-time algorithm \cite{HKM20} and it matches the best known bound for ordinary distance-1 $(\Delta+1)$-coloring in \CONGEST as a function of $n$ alone.
We further explore more efficient algorithms when $\Delta \ll n$. Combining our main method with a range of powerful recent techniques, we obtain an algorithm that runs in time $O(\log \Delta)+2^{O(\sqrt{\log\log n})}$. 

Before discussing our results in more detail, we first discuss why we believe d2-coloring is interesting, what is known for the corresponding coloring problems on $G$ and why it is challenging to transform \CONGEST algorithms to color $G$ into \CONGEST algorithms for d2-coloring. 

Wireless networking is a major motivation for distance-2 coloring, where nodes with a common neighbor should not simultaneously communicate to avoid a collision at the common neighbor~\cite{chlamtackutten85,mecke07}. While the coloring is to be used for scheduling, the wireless channel need not be the medium for \emph{computing} the coloring. With the advent of software-defined radio and hierarchical / heterogeneous networks, it is well motivated to consider coloring computation in a communication model more powerful than radio networks. Yet, asking for the different-message-to/from-all-neighbors feature of \CONGEST may be hoping for too much. 
More generally, we view it as a major question in distributed graph algorithms whether one can relax the communication requirements for graph coloring. We ask:
\begin{quote}
    \emph{How constrained can the communication structure be to allow for fast (logarithmic, sublogarithmic) distributed graph coloring computation?}
\end{quote}

Distributed d2-coloring is an interesting and important problem for several other reasons. The d2-coloring problem for example also occurs naturally when single-round randomized algorithms are derandomized using the method of conditional expectation~\cite{derandomization_FOCS18}. 
d2-coloring in \CONGEST is further of special interest as it appears to lie at the edge of what is computable efficiently, i.e., in polylogarithmic time, while distance-3 coloring is even hard to verify \cite{FHN20}.

Distance-$k$ problems have not been addressed widely in a distributed setting, partly because distance-$k$ communication can be simulated in $k$ steps of the \LOCAL model. In \CONGEST, the situation changes drastically as simulating a single round of a distance-1 coloring algorithm can incur a factor $\Theta(\Delta^{k-1})$ overhead, i.e., even for $k=2$, the overhead can be linear in $\Delta$. Even the very simple algorithm where each node picks a random available color cannot be efficiently used for d2-coloring as it is in general not possible to keep track of the set of colors chosen by $2$-hop neighbors in time $o(\Delta)$.
Recently, Halld\'orsson, Kuhn and Maus \cite{HKM20} treated d2-coloring in \CONGEST and gave a randomized algorithm using $\Delta^2+1$-colors in $O(\log \Delta \cdot \log n)$ rounds, as well as a deterministic algorithm using $(1+\epsilon)\Delta^2$ colors in $poly(\log n)$ rounds. Our main approach builds heavily on their framework, while simplifying certain features and strengthening structural properties. Distributed graph optimization problems on $G^2$ (with \CONGEST-communication in $G$) such as vertex cover and minimum dominating set have recently been studied in \cite{BCMPP20}.

\subparagraph*{Distributed graph coloring}
The standard variant of the distributed coloring problem on $G$ asks for computing a vertex coloring with at most $\Delta+1$ colors, which is computed by a simple sequential greedy algorithm. 
The main focus in the literature on distributed coloring has been on the \LOCAL model, where by now the problem is understood relatively well. 
The best randomized $(\Delta+1)$-coloring algorithm known in the \LOCAL model, due to Chang, Li, and Pettie~\cite{chang18_coloring}, runs in $\poly\log\log n$ rounds. The complexity given in \cite{chang18_coloring} is $2^{O(\sqrt{\log\log n})}$, while the improvement to $\poly\log\log n$ immediately follows from the recent breakthrough work on deterministic \emph{network decomposition} of Rozho\v{n} and Ghaffari~\cite{RG19}.
For a more detailed discussion of related work on distributed coloring, we refer to \cite{barenboimelkin_book,chang18_coloring,kuhn20_coloring}.

While most known distributed coloring algorithms were developed for the \LOCAL model, 
many of them work directly in the \CONGEST model, including those in \cite{alon86,luby86,linial92,johansson99,Kuhn2006On,barenboim10,BEK15,barenboim15,BEG18,kuhn20_coloring,BKM19}. Still, the best complexity known for coloring in \CONGEST, as a function of $n$ alone, is $O(\log n)$, which is achieved by the following very simple \alg{OneShotColoring} algorithm:
Initially all nodes are uncolored. The algorithm runs in synchronous phases, where in each phase, each still uncolored node $v$ chooses a uniform random color among its available colors (i.e., among the colors that have not already been picked by a neighbor) and $v$ keeps the color if no of its uncolored neighbors tries the same color at the same time~\cite{johansson99,BEPS12}.

The only known published algorithm in \CONGEST with a better bound is due to Ghaffari~\cite{ghaffari19}, who obtains a $(\Delta+1)$-coloring in time $O(\log\Delta) + 2^{O(\sqrt{\log\log n})}$. The second term is due to a network decomposition algorithm also introduced in \cite{ghaffari19}. Unlike for results in the \LOCAL model, it is \emph{not} directly possible to replace this decomposition with the recent construction in \cite{RG19} to improve the dependence on $n$. The reason is that the complexity of the network decomposition construction of \cite{RG19} grows at least linearly in the length of the node identifier bit strings. In the \LOCAL model, it is possible to use a standard coloring algorithm of \cite{linial92} to first map the IDs to $O(\log\log n)$-bit values that are unique up to a sufficient distance so that one can afterwards apply the algorithm of \cite{RG19}. Subsequent to the publication of our results \cite{GGR20} improved upon the network decomposition algorithm from \cite{RG19} (to deal with large IDs in the \CONGEST model) and as a result obtains a $O(\log\Delta) + \poly\log\log n$ CONGEST algorithm for $(\Delta+1)$-coloring. Note that if we have graphs of size $N$ and if IDs and colors  can be represented with $\poly\log N$ bits, there is a recent \emph{deterministic} $(\mathit{deg}+1)$-list coloring algorithm running in $\polylog N$ time in \CONGEST \cite{BKM19}.

\subsection{Contributions}

We provide two efficient randomized \CONGEST model algorithms to compute a d2-coloring of a given $n$-node graph $G=(V,E)$. If $\Delta$ is the maximum degree of $G$, the maximum degree of any node in $G^2$ is at most $\Delta + \Delta\cdot(\Delta-1)=\Delta^2$. As a natural analog to studying $(\Delta+1)$-coloring on $G$, we  study the problem of computing a d2-coloring with $\Delta^2+1$ colors. 

\begin{theorem}
\label{thm:d2ColoringRand}
There is a randomized \CONGEST algorithm that d2-colors a graph with $\Delta^2+1$ colors in $O(\log n)$ rounds, with high probability.
\end{theorem}

The algorithm is given in Sec.~\ref{sec:randAlg},
with the key ideas and challenges outlined at the start of the section. Our second algorithm is more efficient if $\Delta\ll n$.

\begin{theorem}
\label{thm:d2ColoringSublog}
There is a randomized \CONGEST algorithm that d2-colors a graph with $\Delta^2+1$ colors in $O(\log \Delta) + 2^{O(\sqrt{\log\log n})}$ rounds, with high probability.
\end{theorem}

\Cref{thm:d2ColoringSublog} relies on the network decomposition algorithm of \cite{Portmann19} that can compute a suitable network decomposition of $G^2$ despite a large ID space in $2^{O(\sqrt{\log\log n})}$ rounds. The unpublished result in \cite{GGR20} computes similar network decompositions despite a large ID space in $\poly\log\log n$ rounds, but only for $G$ and not for $G^2$. If the results of \cite{GGR20} can be extended to $G^2$, which in fact is likely, the runtime of \Cref{thm:d2ColoringSublog} improves to $O(\log \Delta)+\poly\log\log n$, and it then again---as at the time of submission of this manuscript---matches the complexity of ordinary distance-1 coloring in \CONGEST.
The proof of \Cref{thm:d2ColoringSublog} appears in Sec.~\ref{sec:sublog}.




\newcommand{\opsi}{\overline{\psi_v}}

\section{Logarithmic Time Randomized Algorithm}
\label{sec:randAlg}

We give randomized {\CONGEST} algorithms that form a d2-coloring using $\Delta^2+1$ colors. We first introduce notation that we use frequently throughout the proofs in this section.

\subparagraph*{Notation}
The \emph{palette} of available colors is $\{0,1,2,\ldots, \Delta^2\}$. 
The neighbors in $G$ of a node are called \emph{immediate neighbors}, while the neighbors in $G^2$ are \emph{d2-neighbors}.
For a (sub)graph $K$, let $N_K(v)$ denote the set of neighbors of $v$ in $K$, and let $K[v] = K[N_K(v)]$ denote the subgraph induced by these neighbors.
A node is \emph{live} or \emph{uncolored} until it becomes \emph{colored}.
An edge in $G^2$ corresponds to a 2-path (path of length 2) in $G$; thus, $G^2$ can have parallel edges.

A node has \emph{slack $q$} if the number of colors of d2-neighbors plus the number of live d2-neighbors is $\Delta^2+1-q$. In other words, a node has slack $q$ if its palette size is an additive $q$ larger than the number of its uncolored $d2$-neighbors. 

An event holds \emph{w.h.p.}~(with high probability), if for any $c > 0$, we can choose the constants involved so that the event holds with probability $1 - O(n^{-c})$. 

\subsection{Overview}

Our algorithm builds on the approach of \cite{HKM20}, which we first summarize. 
The simple \emph{informed} color guessing approach -- each node tries a random color not used by its d2-neighbors -- fails because the nodes do not have the bandwidth to learn those colors. A simple \emph{uninformed} approach -- trying any random color -- works fine if there is sufficient slack, either because the palette is strictly larger than the degree, or in the beginning when few neighbors have been colored. In this case, even trying a uniformly random color is successful with constant probability. If the node has a \emph{sparse} neighborhood then in the very first round, many pairs of d2-neighbors will conveniently adopt the same color, as proved by Elkin, Pettie and Su \cite{EPS15}, creating the needed slack. We are then left to deal with denser neighborhoods, of varying average non-degree.

The key idea of \cite{HKM20} is to have the colored nodes "help" the \emph{live} nodes by checking random colors on their neighborhoods. This provides a probabilistic filter that helps reduce the load of the live nodes. It turns out that this alone is not sufficient due to \emph{false negatives}: the helper may reject good colors because it has neighbors with those colors. The solution is for the helper to also query one of its neighbor $w$, and forward its color if $w$ is not a d2-neighbor of the live node. It is shown that one of these forms of advice is good with constant probability, but could only argue that for those live nodes with a sparsity in a given range considered. This meant that the round complexity of the method had an extra $O(\log \Delta)$ factor for ranging through the different sparsity levels, on top of the $\log n$ factor for finishing off all nodes of that sparsity.

The main technical ingredient behind our $O(\log n)$-round algorithm is the adaptation and extension of the \emph{almost-clique decomposition} (ACD) method initially proposed by Harris, Schneider and Su \cite{HSS16} for the \LOCAL model and expanded by Assadi, Lee and Khanna \cite{ACK19} for streaming and massively parallel settings. 
The nodes are partitioned into a set of sparse nodes-- which can be handled by uninformed guesses -- and low-diameter clusters of dense nodes. The ACD achieves the same aims as the \emph{similarity graphs} of \cite{HKM20} that guide the querying and ensure effective filtering, but attain some additional crucial properties such as near-regular high degree. Our extension to ACD is to ensure that all nodes outside clusters have a low degree into the cores of the clusters, strengthening the divide between inside and outside. The decomposition additionally simplifies the technical arguments, including load balancing and probabilistic independence. The key property that we then obtain is that in each iteration, every live node (with at least logarithmic size palette) becomes colored with constant probability.
That makes even faster algorithms possible, as we show in the next section. To finish off the nodes with a palette of at most logarithmic size, we apply a second method of \cite{HKM20} black-box, which \emph{learns} the palette of the live nodes and then performs informed color guessing.


\subsection{Algorithm Description}

We now outline our algorithm, followed by details on the implementation. 

Each live node $v$ repeatedly \emph{tries} a suggested color, which means to first \emph{validate} it and then \emph{contest} it. Validating a color means sending it to all immediate neighbors, who then report back if they or any of their neighbors had already adopted that color. Contesting a validated color means proposing it to intermediate neighbors, who report back if any other node also proposes it.
If all answers are negative, then $v$ adopts the color. 

In what follows, let $\epsilon = 1/60$, $c_0 = 48 e^4/\epsilon^2$, and $c_3$ be a constant to be determined. Also, $c_2$ is a sufficiently large constant needed for concentration.

\begin{quote}
   \textbf{Algorithm} \alg{d2-Color}

   If $\Delta^2 \ge c_2\log n$ then  \\
   \hspace*{2em} {\bf 1.} Compute an almost-clique decomposition. \\
   \hspace*{2em} {\bf 2.} repeat $c_0 \log n$ times: \\
   \hspace*{4em} Each live node picks a random color and \emph{tries} it. \\
   \hspace*{2em} {\bf 3.} repeat $c_3\log n$ times \\
    \hspace*{4em} \alg{Reduce-Phase}() \\
   \hspace*{2em}  {\bf 4.} \alg{LearnPalette}() \\
    \alg{FinishColoring()}
\end{quote}

We will discuss and analyze Steps 1--3 of the above algorithm in detail in the following. The remaining steps, \alg{LearnPalette}() and \alg{FinishColoring}(), are from \cite{HKM20}. In  \alg{LearnPalette}(), each live node \emph{learns the palette} of still available colors by cooperatively tallying the colors of d2-neighbors. In \alg{FinishColoring}(), the maximum degree is sufficiently small so that we can efficiently simulate the classic algorithm of informed color guessing to color the remaining live nodes. 

The first step of the algorithm is to compute a decomposition of the nodes into: a) a set of nodes inducing a subgraph  (in $G^2$) that is sufficiently sparse and b) a disjoint collection of almost-cliques (also in $G^2$). In the following, each of the almost-cliques is called a component $C$ and we use two graphs $H$ and $\hat{H}$ (closely related to the ones in \cite{HKM20}) that both essentially consist of all the components (i.e., all the almost-cliques) and all the $G^2$-edges connecting two nodes in the same component. Also computed within each component is a spanning tree for a fast aggregation.
The exact definitions of the decomposition and of the graphs $H$ and $\hat{H}$ appear in Subsection \ref{ssec:similarity}. 

We next detail the steps of \alg{Reduce-Phase}(), which is the core piece of our algorithm. 
\medskip

\noindent\textbf{Algorithm} \alg{Reduce-Phase}()
\begin{enumerate}
\itemsep 0em 
  \item Each live node randomly decides to be \emph{active} with probability $1/8$. All other nodes are \emph{inactive}.
  \item Compute $\phi$, the number of active live nodes in the component $C$, and distribute it to the nodes of $C$.
  \item Each inactive node $u \in C$ computes $\phi_u$, the number of 2-paths to active nodes (by asking its immediate neighbors of their active immediate neighbors). $u$ flips a biased coin: with probability $\min(1,\phi_u/(4\phi))$ it picks one of the $\phi_u$ paths uniformly at random, while with probability $\max(0,1-\phi_u/(4\phi))$, $u$ stops the execution of this iteration. Let $v$ denote the active node at the other end of the path chosen. $u$ verifies that it has only one 2-path to $u$ (by inquiring to its immediate neighbors), and otherwise stops execution of this iteration.
  \item $u$ picks a random color $\hat{c}$ different from its own. If that color is not used by any of its $\hat{H}$-neighbors, then $u$ sends the color to $v$ as a proposal, assigning it a uniformly random priority. $v$ tries the proposed color of highest priority (if any).
  \item $u$ sends query $(v,u)$ along a random 2-path to an inactive $\hat{H}$-neighbor $w$, and assigns it a random priority.
  \item Upon receipt of a query, node $w$ selects the highest priority query $(v,u)$, checks if $v$ is a d2-neighbor, and if $v$ and $w$ are not d2-neighbors, it sends its color $c(w)$ to $v$ (through $u$).
  \item The active node $v$ tries a color chosen uniformly random among the received proposed colors from Step 6 (if any).
  \end{enumerate}

A colored node assists an active node $v$ in two ways: a) guesses and validates a random color for $v$ to try, and b) sees if a random d2-neighbor is also a d2-neighbor of $v$. This is a probabilistic filter that reduces the workload of the active nodes. The key idea is that one of these forms of assistance is likely to be successful.

\mypar{Complexity}
We discuss the almost-clique decomposition in the next subsection and show how to implement it in $O(\log n)$ rounds, w.h.p.
The second step clearly takes $\Theta(\log n)$ rounds.
The procedure \alg{Reduce-Phase} takes 24 rounds, or 8 (Step 2), 2 (Step 3), 2 (Step 4), 2 (Step 5), 6 (Step 6), and 4 (Step 7, including the notification of a new color).

Outline of this section: In Sec.~\ref{ssec:similarity} we describe the almost-clique decomposition and derive key structural properties of dense subgraphs, and in Sec.~\ref{ssec:correctness} we prove the correctness of the algorithms, i.e., we show that any node is colored w.h.p. after $O(\log n)$ rounds.

\subsection{Almost-Clique Decomposition}
\label{ssec:similarity}

We next define the notion of local sparsity and the almost-clique decomposition that we use in our paper. The first definition is a slight adaptation of a similar definition in \cite{EPS15}.

\begin{definition}
A node $v$ is \emph{$\zeta$-sparse} (or \emph{has sparsity} $\zeta$) if $G^2[v]$ contains $\binom{\Delta^2}{2} - \Delta^2\cdot \zeta$ (distinct) edges. 
\label{D:solid}
\end{definition}

Sparsity is a rational number that indicates how many edges are missing from $G^2[v]$, compared with the densest case (when $v$'s d2-neighborhood is a $\Delta^2$-clique). If no pairs of d2-neighbors of $v$ are adjacent, then $\zeta = (\Delta^2-1)/2$, while if $G^2[v]$ forms a $\Delta^2$-clique, then $\zeta=0$.
Elkin, Pettie and Su \cite{EPS15} formalized the connection between sparsity and slack that appears after trying one uniformly random color.

\begin{proposition}[\cite{EPS15}, Lemma 3.1]
Let $v$ be a vertex of sparsity $\zeta$ and let $Z$ be the slack of $v$ after trying a single random color.
Then, 
 $\Pr[Z \le \zeta/(4 e^3)] \le e^{-\Omega(\zeta)}$.
\label{P:sparsity}
\end{proposition}

We require the constant $c_2$ to be such that if $\zeta \ge c_2\log n$, then the contrapositive of Prop.~\ref{P:sparsity} yields that $Z \ge \zeta/(4e^3)$, w.h.p.

\mypar{Decomposition}
We adapt the almost-clique decomposition of \cite{ACK19} (building on \cite{HSS16}) for the distance-2 setting in \CONGEST and endow it with an additional property. 

\begin{definition}
Assume $\epsilon \le 1/60$.
Nodes $u$ and $v$ are $\epsilon$-\emph{similar} if they share at least $(1-\epsilon)\Delta^2$ common d2-neighbors. 
An \emph{almost-clique decomposition} (ACD) with parameter $\epsilon$ is a collection of sets $V_\ast, \hat{C}_1, \hat{C}_2, \ldots, \hat{C}_k$ that cover $V$ and where the $\hat{C}_i$ are disjoint.
Denote $C_i = \hat{C_i}\setminus V_\ast$, for $i=1,\ldots, k$.
The decomposition satisfies the following properties:
\begin{enumerate}
    \item The nodes in $V_\ast$ have sparsity at least $\epsilon^2 \Delta^2/4$.
    \item For any $i \in [k]$, $C_i$ and $\hat{C}_i$ satisfy:
\begin{enumerate}
    \item $|C_i| \ge (1-2\epsilon) \Delta^2$.
    \item The nodes in $\hat{C}_i$ are mutually $10\epsilon$-similar. 
    \item Each $v \in \hat{C}_i$ has at most $28\epsilon\Delta^2$ d2-non-neighbors in $\hat{C}_i$ (i.e., $|\hat{C}_i \setminus N_{\hat{C}_i}(v)| \le 28\epsilon\Delta^2$).
    \item Each $v \in \hat{C}_i$ has at least $(1-10\epsilon)\Delta^2$ d2-neighbors in $C_i$. 
    \item Each $v \in C_i$ is $\epsilon$-dissimilar to every node outside $\hat{C_i}$.
\end{enumerate}
\end{enumerate}
\label{D:acd}
\end{definition}

We refer to each $C_i$ as a \emph{component} and $\hat{C}_i$ as an \emph{extended component}.
The properties imply additional ones: Each extended component is of size at most $(1+28\epsilon)\Delta^2$; and any two nodes in an extended component are within two hops (in $G^2$). The additional property we need that is not in the formulations of \cite{ACK19} or \cite{HSS16} is Property 2(e).

Let $H$ denote the subgraph of $G^2$ induced by the components $C_1, \ldots, C_k$, i.e., $H = \cup_i G^2[C_i] = (V\setminus V_\ast, E_H)$ where $E_H$ consists of the pairs of d2-neighbors within the same component. Similarly, let $\hat{H} = \cup_i G^2[\hat{C}_i]$. 
We consider $H$, $\hat{H}$ and $G^2$ to be \emph{simple} graphs, ignoring multiple 2-paths between the same pair of nodes.

\begin{lemma}
There is an $O(\log n)$-round {\CONGEST} algorithm to form an almost-clique decomposition, for any fixed $\epsilon > 0$.
Afterwards, each node knows its component ID.
\label{lem:acd}
\end{lemma}

It is somewhat surprising that such a decomposition can be established efficiently in \CONGEST. The key implementation ideas are in \cite{ACK19} for other models, which are essentially based on randomly sampling nodes. In the distance-2 setting we have the additional challenge of communication with one's d2-neighbors, but the key is to have both parties communicate only with the intermediate node that makes the deciding. 
The proofs are given in Appendix~\ref{app:acd}.

We strengthen the ACD-properties for dense nodes and show that they scale with the node sparsity.
Note that dense nodes can have non-trivial sparsity and  it is crucial in our argument to leverage the corresponding slack. 

\begin{lemma}
Let $\epsilon \le 1/30$.
Let $v$ be a node of sparsity $\zeta$ in an almost-clique $C$ and extended component $\hat{C}$.
Then,
\begin{enumerate}
    \item $v$ has at least $\Delta^2 - (2\zeta+1)/\epsilon$ $\hat{H}$-neighbors (in $\hat{C})$,
    \item $v$ has at most $|\hat{C} \setminus N_{G^2}(v)| \le 3\zeta$ $\hat{H}$-non-neighbors, and
    \item The number of edges in $\hat{H}[v]$ is at least $|E(\hat{H}[v])| \ge \binom{\Delta^2}{2} - (2/\epsilon + 1)\zeta\Delta^2$.
\end{enumerate}
\label{L:h-degree}
\end{lemma}

\begin{proof}  
Recall that by the definition of sparsity, $G^2[v]$ has exactly $\Delta^2((\Delta^2-1)/2 - \zeta)$ edges.

\textbf{1}. 
A d2-neighbor of $v$ that is not $\epsilon$-similar to $v$ can share at most $(1-\epsilon) \Delta^2$ common d2-neighbors with $v$ by ACD property 2(e).
In other words, the d2-neighbors of $v$ that are not $\hat{H}$-neighbors can have degree at most 
$(1-\epsilon)\Delta^2$ in $G^2[v]$.
The number of edges in $G^2[v]$ is then at most
\[ \frac{1}{2}\left( |N_{\hat{H}}(v)| \Delta^2 + (|N_{G^2}(v)| - |N_{\hat{H}}(v)|) (1-\epsilon) \Delta^2\right) \le
\frac{\Delta^2}{2} \left((1-\epsilon) \Delta^2 + \epsilon|N_{\hat{H}}(v)| \right) \ . \]
Combining the two bounds on the number of edges in $G^2[v]$, 
\[ \epsilon |N_{\hat{H}}(v)| \ge \Delta^2 - 1 - 2\zeta - (1-\epsilon)\Delta^2 
  = \epsilon \Delta^2  - 1 - 2\zeta\ . \]
Namely, the number of $\hat{H}$-neighbors of $v$ is at least $\Delta^2 - (2\zeta+1)/\epsilon$.

\textbf{2}. By sparsity, there are at most $(2\zeta+1) \Delta^2$ edges of $\hat{H}$ with exactly one endpoint in $N_{G^2}(v)$. Nodes in $\hat{C} \setminus N_{G^2}(v)$ share at least $(1-10\epsilon)\Delta^2$ d2-neighbors with $v$, by ACD property 2(b). Thus, there are at most $2\zeta \Delta^2/((1-10)\epsilon \Delta^2) = 2\zeta/(1-10\epsilon) \le 3\zeta$ nodes in $\hat{C}$ that are not d2-neighbors of $v$, using that $\epsilon \le 1/30$.

\textbf{3}.
By \textbf{1} of this lemma, $v$ has degree at least $\Delta^2 - q$ in $\hat{H}$, where $q = (2\zeta+1)/\epsilon$. The at most $q$ nodes in $N_{G^2}(v) \setminus N_{\hat{H}}(v)$ have degree sum at most $q (\Delta^2-q)$. Thus, the number of edges in $\hat{H}[v]$ is at least $\binom{\Delta^2}{2} - \zeta\Delta^2 - q\Delta^2$.
\end{proof}

\subsection{Correctness}
\label{ssec:correctness}
We prove that \alg{d2-Color} correctly d2-colors $G$ with $\Delta^2+1$ colors in $O(\log n)$ rounds.
We assume that the almost-clique decomposition and the graphs $H$ and $\hat{H}$ have been correctly constructed, in the sense of Def.~\ref{D:acd}.
Also, that nodes of sparsity $\zeta \ge c_2 \log n$ have slack at least $\zeta/(4e^3)$ as promised by Prop.~\ref{P:sparsity}.
All statements in this section are conditioned on these events.

We first give a high-level proof which encapsulates the core of the technical argument in the following lemma, which is then proven in the upcoming subsubsection.

\begin{lemma}
There is an absolute constant $c'$ such that the following holds.
For a live node $v$ in given iteration of \alg{Reduce-Phase}, 
there is a subset $S \subseteq \psi_v$ of size at least $|\psi_v|/2$ such that each color in $S$ has probability at least $1/(c' |\psi_v|)$ of being validated by $v$. 
\label{L:xxx}
\end{lemma}

We then easily dispose of the sparse nodes. Since they have slack linear in their degree, they get colored with constant probability in each round, simply by contesting a uniformly random color.

\begin{lemma}
Every node in $V_\ast$ is colored after Step 2 of \alg{d2-Color}, w.h.p.
\label{O:sparse}
\end{lemma}

\begin{proof}
Let $v \in V_\ast$. By Def.~\ref{D:acd}(1), $v$ has sparsity at least $\zeta \ge \epsilon^2 \Delta^2/4$, and by Prop.~\ref{P:sparsity}, it has slack at least $c_{13} \doteq \epsilon^2/(16e^3)$, w.h.p. 
Furthermore, the probability that no d2-neighbor of $v$ tries the same color in the same round is at least $(1-1/(\Delta^2+1))^{\Delta^2} \ge 1/e$, applying Ineq.~(\ref{eq:inv-e}). Thus,
with probability at least $c_{13}/e$, $v$ becomes colored in that round. Hence, the probability that it is not colored in all $c_0 \log n$ rounds is at most $(1-c_{13}/e)^{c_0\log n} \le e^{-c_0 c_{13}/e \log n} \le n^{-c_0 c_{13}/e} = n^{-3}$, since $c_0 = 48 e^4/\epsilon^2 = 3e/c_{13}$.
\end{proof}

\medskip

\noindent\textbf{Theorem \ref{thm:d2ColoringRand} (restated)}. 
\emph{There is a randomized {\CONGEST} algorithm to d2-color with $\Delta^2+1$ colors in $O(\log n)$ rounds,  with high probability.}
\begin{proof}
By Lemma \ref{O:sparse}, it suffices to focus on the dense nodes.
We first claim that in each iteration, each live node $v$ with palette size $\Omega(\log n)$ becomes colored with a constant non-zero probability.

Consider a given iteration and a live node $v$. 
With probability 1/8, $v$ is active. It has at most $|\psi_v|$ live neighbors and expected at most $|\psi_v|/8$ are active. By Markov's inequality, at most $|\psi_v|/4$ are active, with probability at least $1/2$. By Lemma \ref{L:xxx}, there is a subset $S \subseteq \psi_v$ of size at least $|\psi_v|/2$ such that each color in $S$ has probability at least $1/(c' |\psi_v|)$ of being validated. 
Independent of what these active neighbors choose, there is then a subset of at least $|\psi_v|/2 - |\psi_v|/4 = |\psi_v|/4$ colors that are available to $v$, i.e., are not contested by d2-neighbors of $v$ in that iteration.
The probability that one of them is validated, and leading to a valid coloring of $v$, is then at least
\[ c_* = \frac{1}{8} \cdot \frac{1}{2} \cdot \frac{|\psi_v|/4}{c'|\psi_v|} = \frac{1}{64 c'}\ , \]
establishing the claim.

Applying Chernoff bound (\ref{eq:chernoff-upper}) to the above claim, after $5/c_* \cdot \log n$ iterations of \alg{Reduce-Phase}, it holds with probability at least $1-1/n^3$ that all nodes are either colored or have palette size $O(\log n)$ (in which case they have $O(\log n)$ uncolored d2-neighbors). 
The coloring is then completed by the two algorithms of \cite{HKM20}, both running in $O(\log n)$ rounds.
\end{proof}

\subsubsection{Proof of Lemma \ref{L:xxx}}

We prove our main result in two parts, given in 
Lemmas \ref{L:phi-ophi} and \ref{L:ophi}, distinguishing between the two forms of making progress: based on Step 4 or Steps 5-7 of \alg{Reduce-Phase}. 

\begin{definition}
An inactive node is \emph{$v$-decent} (or just \emph{decent}) if it has at most $4\phi$ 2-paths in its almost-clique $C$ to active nodes (in $C$) \emph{and} has exactly one 2-path to $v$.
\label{D:decent}
\end{definition}
The distinction between 2-paths and d2-neighbor relations is the rationale for the \emph{decent} definition. Those nodes with lots of paths to active nodes can cause much congestion with poor proposals, while being of limited use to those $H$-neighbors to which they have few paths. 

\begin{lemma}
Let $v$ be a live node and $w$ be a node, both in $\hat{C}$.
Then, $v$ and $w$ have at least $\Delta^2/4$ common d2-neighbors in $C$ that are $v$-decent.
\label{L:decent}
\end{lemma}

\begin{proof}
The two nodes $v$ and $w$ are $10\epsilon$-similar, by Def.~\ref{D:acd}(2b).
Since $v$ has at least $(1-10\epsilon)\Delta^2$ distinct d2-neighbors in $C$ (by Def.~\ref{D:acd}(2d)), they share at least $(1-20\epsilon)\Delta^2 \ge 2\Delta^2/3$ d2-neighbors in $C$ (using that $\epsilon \le 1/60$).
This also means that there are at most $10\epsilon\Delta^2 \le \Delta^2/6$ nodes in $C$ with multiple 2-paths to $v$.
Also, since there are at most $\phi\Delta^2$ total number of 2-paths to the $\phi$ live nodes in $C$, there are at most $\Delta^2/4$ nodes with $4\phi$ or more 2-paths to live nodes. 
Hence, there are at least $2\Delta^2/3 - \Delta^2/6  - \Delta^2/4 = \Delta^2/4$ common d2-neighbors of $v$ and $w$ in $C$ that are decent, i.e., have at most $4\phi$ 2-paths to active nodes and exactly one 2-path to $v$. 
\end{proof}

A color proposed to an active node $v$ is \emph{bad} if it is already assigned to a d2-neighbor of $v$. Namely, it is bad if it is a "false positive".

\begin{lemma}
The expected number of bad proposals generated in Step 4 for $v$ is at most $(1/\epsilon+1)\zeta/\phi$.
\label{L:total-proposals-step3}
\end{lemma}
\begin{proof}
Let $u$ be an inactive $H$-neighbor of $v$.
Let $Y_u$ be the event that $u$ picks $v$ in Step 2, and note that $\Pr[Y_P] \le 1/(4\phi)$.
Let $X_u$ be the event that $u$ generates a bad proposal for $v$ in Step 4. That event occurs when $u$'s randomly chosen color is used by a node in $S_u$, where $S_u = N_{G^2}(v)\setminus N_{\hat{H}}[u]$ is the set of d2-neighbors of $v$ that are not $\hat{H}$-neighbors of $u$ (nor $u$ itself). The number of such colors is at most $|S_u| = |N_{G^2}(v) \setminus N_{\hat{H}}[u]| \le (\Delta^2-1) - |N_{\hat{H}}(u)\cap N_{\hat{H}}(v)|$. 
There are at most $\Delta^2$ colors to choose from -- all except the one on $u$ -- so 
\[ \Pr[X_u | Y_u] \le \frac{|S_u|}{\Delta^2} \le \frac{(\Delta^2-1) - |N_{\hat{H}}(u)\cap N_{\hat{H}}(v)|}{\Delta^2} \ . \]
By applying Lemma \ref{L:h-degree}(3), we have that
\[ \sum_{u \in N_{\hat{H}}(v)} |N_{\hat{H}}(u)\cap N_{\hat{H}}(v)| = 2 |E(\hat{H}[v])| \ge \Delta^2(\Delta^2-1) - (4/\epsilon+2)\zeta\Delta^2\ . \]
Combining the two bounds, letting $I$ denote the set of inactive $H$-neighbors of $v$, we get that 
\[ \sum_{u \in I} \Pr[X_u | Y_u] \le \sum_{u \in N_{\hat{H}}(v)} \Pr[X_u | Y_u] \le (4/\epsilon+2)\zeta \ . \]
Hence, the expected number of bad proposals generated for $v$ is
\[ \sum_{u \in I} \Pr[X_u \cap Y_u] = \sum_{u \in I} \Pr[Y_u] \cdot \Pr[X_u | Y_u] 
 \le \frac{1}{4\phi} \sum_{u \in I} \Pr[X_u] \le \frac{(4/\epsilon + 4)\zeta}{4\phi}\ . \qedhere\]
\end{proof}

Let $\psi_v$ denote the set of colors in $v$'s palette before a given round, i.e., the set of colors that have not already been taken by its d2-neighbors. Let $\opsi$ be the set of colors in $v$'s palette that appear on nodes in $\hat{C}$. These colors must then appear only on non-$\hat{H}$-neighbors of $v$.

\begin{lemma}
Suppose $|\psi_v| \ge 2 |\opsi|$ and $|\psi_v| = \Omega(\log n)$.
Then, there is an absolute constant $c$ such that each color in $\psi_v \setminus \opsi$ has probability at least $1/(c|\psi_v|)$ of being validated and contested by $v$ in Step 4.
\label{L:phi-ophi}
\end{lemma}

\begin{proof}
Let $\hat{\psi} = \psi_v \setminus \opsi$.
Any color from $\hat{\psi}$ that is guessed in Step 4 (by some $H$-neighbor $u$ of $v$) becomes a \emph{good} proposal to $v$ (i.e., one that would pass validation). 
Let $A$ ($B$) denote the expected number of good (bad) proposals to $v$, respectively.
Let $q$ be a color in $\hat{\psi}$ and let $A_q$ be the expected number of proposals of $q$ to $v$. 
We shall show that $A_q$ is large, for colors in $\hat{\psi}$, and thus $A$ is large in comparison to $B$. We then show that $A_q$ is also large relative to the total number of proposals, $A+B$. 

The probability that a decent $H$-neighbor $u$ chooses to help $v$ is $1/(4\phi)$, and the probability that it guesses $q$ is $1/\Delta^2$.
By Lemma \ref{L:decent}, $v$ has at least $\Delta^2/4$ decent $H$-neighbors.
Summing up, $A_q \ge \sum_u 1/(4\phi) \cdot 1/\Delta^2 \ge 1/(16\phi)$, and $A \ge \sum_{q \in \hat{\psi}} A_q \ge |\hat{\psi}_v|/(16\phi) \ge |\psi_v|/(32\phi)$.
By Lemma \ref{L:total-proposals-step3}, $B \le (1/\epsilon + 1)\zeta/\phi$ and by Prop.~\ref{P:sparsity}, $|\psi_v| \ge \zeta/(4e^3)$.
Thus, $B \le (128 e^3(1/\epsilon+1)) A$.
We can also bound $A$ from above, summing over the at most $\Delta^2$ $H$-neighbors and all the colors in $v$'s palette: 
\[ A \le \sum_{q' \in \psi_v} \sum_{u \in N_H(v)} \frac{1}{4\phi\Delta^2} = \frac{|\psi_v|}{4\phi} \le 4|\psi_v| A_q\ .  \]

By Markov's inequality, the probability that at most $2(A+B)$ proposals are generated for $v$ is at least $1/2$. The probability that a proposal of $q$ is chosen for validation is then at least
\[ \frac{A_q}{4(A+B)} \ge \frac{A/(4|\psi_v|)}{4(1+128e^3(1/\epsilon+1))A} = \frac{1}{16(1+128e^3(1/\epsilon+1))|\psi_v|}\ . \qedhere\]
\end{proof}

\begin{lemma}
Suppose $|\psi_v| < 2 |\opsi|$.
Then, there is an absolute constant $c$ such that each color in $\opsi$ has probability at least $1/(c|\psi_v|)$ of being validated and contested by $v$ in Step 7.
\label{L:ophi}
\end{lemma}

\begin{proof}
Only colors of nodes in $\overline{N_{\hat{C}}(v)}=\hat{C}\setminus N_{G^2}(v)$ (nodes in the extended components that are not d2-neighbors of $v$) have potential to become proposed to $v$ in Steps 5-6, i.e., the colors in $\overline{\psi_v}$.
Let $q$ be a color in $\overline{\psi_v}$ and let $w$ be a node in $\overline{\hat{C}}[v]$ of that color.

Let $u$ be a decent $H$-neighbor of $v$ that is also a $\hat{H}$-neighbor of $w$.
Let $A_{uw}$ be the event that $u$ chooses to help $v$ and that it picks $w$ (in Step 5). Since these picks are independent, $\Pr[A_{uw}] \ge 1/(4\phi\Delta^2)$.
Let $B_{uw}$ be the event that $u$'s proposal of $w$'s color becomes validated in Step 7. In addition to $A_{uw}$ holding, $B_{uw}$ additionally requires that the proposal survives the culling at $w$ (in Step 6) and at $v$ (in Step 7). Most of the rest of the proof is focused on bounding this probability.

The expected number of queries that $w$ receives is at most $5/4$, since $w$ has at most $\Delta^2$ 2-paths to nodes in $C$ and each of them has at least $(1-10\epsilon)\Delta^2 \ge 4\Delta^2/5$ $H$-neighbors (by Def.~\ref{D:acd}(2d)).
Then, by Markov's inequality, $w$ receives at most 5 queries, with probability at least $3/4$.

We next bound the expected load on $v$.
Let $w' \in \overline{N_{\hat{C}}(v)}=\hat{C}\setminus N_{G^2}(v)$. For each 2-path $P$ from $w'$ to a decent node $u_P$ in $C$, let $X_P$ be the event that the color of $w'$ is forwarded to $v$ through $P$. This is the product of two independent events: $Y_P$, that $u_P$ makes contact with $v$ in Step 3, and $Z_P$, that $u_P$ forwards a query along $P$ in Step 5. Since $u_P$ has at least $(1-10\epsilon)\Delta^2 > 5\Delta^2/6$ distinct d2-neighbors, the probability of $Z_P$ is at most $6/(5\Delta^2)$. Also, by the constraints in Step 2, the probability of $Y_P$ is $1/(4\phi)$. Thus, $\Pr[X_P] = \Pr[Y_P] \cdot \Pr[Z_P] \le 1/(4\phi) \cdot 6/(5\Delta^2) = 3/(10\phi \Delta^2)$.
Summing up over all the at most $\Delta^2$ paths from $w'$, the probability that a proposal from $w'$ arrives at $v$ is at most $3/(10\phi)$. Summing up over all the at most $3\zeta$ nodes $w'$ in $\overline{N_{\hat{C}}(v)}$, the expected number of proposals headed for $v$ in Step 7 is at most $9\zeta/(10\phi) \le \zeta/\phi$. 
By Markov's inequality, the probability that $v$ receives more than $4\zeta/\phi$ proposals is at most $1/4$. 

We can now combine the two bounds: with probability at least $1/2$, $w$ receives at most $5$ queries in a given iteration and
$v$ receives at most $4\zeta/\phi$ proposals.
Then, conditioned on $A_{uw}$, $w$ receives at most $6$ queries and $v$ receives at most $4\zeta/\phi+1$ proposals, with probability at least $1/2$. 
Hence, given $A_{uw}$, the query from $u$ becomes validated with probability 
\[ \Pr[B_{uw} | A_{uw}] \ge \frac{1}{2} \cdot \frac{1}{6}\cdot \frac{1}{4\zeta/\phi+1}\ . \]
Hence, the query of $u$ (from $w$ to $v$) becomes validated with probability
\[ \Pr[B_{uw}] = \Pr[A_{uw}] \cdot \Pr[B_{uw}| A_{uw}] \ge \frac{1}{4\phi\Delta^2} \cdot \frac{1}{12(4\zeta/\phi+1)} = \frac{1}{48(4\zeta + \phi)\Delta^2} \ . \]
By Lemma \ref{L:decent}, the set $D$ of 
decent $H$-neighbors of $v$ that are also $\hat{H}$-neighbors of $w$ is of size at least $\Delta^2/4$.
As the events $B_{uw}$ are disjoint, the probability that $v$ validates a proposal of the color $q$ (from $w$) is at least $1/(200(4\zeta + \phi))$.
Observe that the palette of $v$ has size at least the number of live d2-neighbors in $\hat{C}$ and at least the slack promised by the sparsity of $v$. Thus, $|\psi_v| \ge \max(\phi-3\zeta,\zeta/(4e^3)) \ge (\phi+4\zeta)/(20e^3)$.
Hence, the probability that $v$ validates an arbitrary color $q \in \opsi$ is at least $1/(4000e^3|\psi_v|)$.
\end{proof}


Lemma \ref{L:xxx} follows from Lemmas \ref{L:phi-ophi} and \ref{L:ophi}.

\section{Sub-Logarithmic Distance-2 Coloring}
\label{sec:sublog}

In this section, we extend the algorithm of Section \ref{sec:randAlg} and combine it with the \emph{graph shattering} technique~\cite{BEPS12,Beck91}, which has been used extensively in recent years to get sub-logarithmic-time distributed algorithms for a large number of graph problems (mostly in the \LOCAL model). By using this technique in our setting, we prove the following theorem.

\medskip

\noindent\textbf{Theorem \ref{thm:d2ColoringSublog} (restated)}. 
\emph{There is a randomized \CONGEST algorithm that d2-colors a graph with $\Delta^2+1$ colors in $O(\log \Delta) + \ND\cdot\poly\log\log n$ rounds, with high probability.}

\smallskip
Here $\ND$ is the sum of $d\cdot c \cdot x$ and the time to compute a distance-$2$ $x$-CONGEST-routable network decomposition with weak cluster diameter $d$ and $c$ cluster colors on subgraphs of size $\poly\log n$ with node identifiers from a space of size $\poly n$ (cf. \Cref{def:decomposition} for a formal definition).
\begin{remark}  
The current state of the art  for $\ND$ is $2^{O(\sqrt{\log\log n})}$ \cite{Portmann19}. However, the complexity for distance-$1$ network decompositions that can deal with a large identifier space was improved subsequent to the submission of this manuscript to $\poly\log\log n$ rounds \cite{GGR20}.  Before the publication of \cite{GGR20} the complexity in \Cref{thm:d2ColoringSublog} for distance-$2$ coloring matched the state of the art for distance-$1$ $(\Delta+1)$-coloring~\cite{ghaffari19}. As the achievements of \cite{GGR20} improve the complexity for distance-$1$ coloring from $O(\log \Delta)+2^{O(\sqrt{\log\log n})}$ to $O(\log \Delta)+\poly\log\log n$ there currently is a gap between the complexities of distance-$1$ and distance-$2$ coloring. If \cite{GGR20} (or an alternative approach) extends to distance-$2$ decompositions, and such an extension is very likely, it will match again. In the remaining part of the writeup we use the best known upper bound of $\ND=2^{O(\sqrt{\log\log n})}$.
\end{remark}

From a very high-level point of view, the rough idea of graph shattering applied to our problem is as follows. The algorithm of Section \ref{sec:randAlg} consists of $O(\log n)$ individual $O(1)$-round steps, where in each step, each live node gets colored with constant probability. Thus, very roughly, if we just run the algorithm for $O(\log \Delta)$ steps, each node remains uncolored with probability at most $1/\poly(\Delta)$. Further, if nodes succeeded sufficiently independently, after $O(\log\Delta)$ rounds, each node would only have $O(\log n)$ uncolored neighbors. By combining these two properties, one can hope that after $O(\log\Delta)$ rounds, all the remaining live nodes induce components (in $G^2$) of size at most $\polylog n$. By adapting techniques developed in \cite{BEPS12} to our $G^2$-coloring algorithm, we will show that this indeed (almost) is the case. We call this part of the algorithm, where we reduce the original problem to a problem on components of $\polylog n$ size, the \emph{preshattering phase} of our algorithm.

The remaining problem that we need to solve on the components of size $\polylog n$ is a list coloring problem. Because these problems for each component are on much smaller graphs, they can be solved efficiently by using the best known deterministic algorithm. For the specific setting, where we have small components, but each node still has an ID from the original large ID space, the best known deterministic \CONGEST algorithm (that can tolerate such a large ID space and works for $G^2$) can be obtained by combining a network decomposition algorithm of Ghaffari and Portmann~\cite{Portmann19} with a recent deterministic \CONGEST coloring algorithm of Bamberger, Kuhn, and Maus~\cite{BKM19}. It requires $2^{O(\sqrt{\log N})}=2^{O(\sqrt{\log\log n})}$ time, where $N=\polylog n$ is the maximum component size. We call this second phase of solving the remaining list coloring instances on the components the \emph{postshattering phase}.

While the general outline of the algorithm is relatively standard and largely follows the ideas of the distance-1 coloring algorithm for the \LOCAL model in \cite{BEPS12}, there are various challenges that we have to cope with in order to apply the idea in the \CONGEST model and to the d2-coloring problem.
In \cite{BEPS12,ghaffari19},
the algorithm for the preshattering phase is very simple: In each step of the algorithm, every live node tries a uniformly random color from its current list of available colors. As we have seen in Sec.~\ref{sec:randAlg}, we cannot run this algorithm in the d2-coloring setting as it is not possible for a live node to learn its list of available colors (i.e., learn the colors already chosen by its 2-neighbors). We would therefore like to show that the much more involved randomized algorithm of Sec.~\ref{sec:randAlg} also has the same shattering properties as the basic ``choose-a-random-available-color'' algorithm. Unfortunately, this is not obvious and we use a multi-stage algorithm to prove what we need. Greatly simplified, we do the following. We first show that $O(\log\Delta)$ rounds of an adaptation of the algorithm of Sec.~\ref{sec:randAlg} suffice to (essentially) reduce the maximum degree of the subgraph of $G^2$ induced by the live nodes to $O(\log n)$. At this point, it is possible for each live node to learn a sufficiently large list of available colors in $O(\log\Delta)$ rounds and we can now indeed run the basic preshattering algorithm of \cite{BEPS12} to reduce the problem to a problem on $\polylog n$-size components.

For the postshattering phase, while we only have components of $\poly\log n$ size, the input to the problem is still large because each node still has an ID of size $O(\log n)$ bits and because each node has a color list consisting of up to $O(\log n)$ colors from a range of size $O(\Delta^2)$. In order to have an efficient \CONGEST algorithm for the problem, we have to reduce both the ID space and the color space of the remaining components. It is sufficient to obtain new node IDs that are unique up to distance $\poly\log\log n$. We can obtain such IDs with $O(\log\log n)$ bits by first applying the network decomposition algorithm of  \cite{Portmann19} and then assigning unique labels in each cluster. For reducing the color space, we show that in each cluster of the network decomposition, we can efficiently (and deterministically) find a renaming of the colors such that for every node $v$, all colors in $v$'s list are mapped to distinct new colors and such that the colors are from a space of size $\poly\log n$. For each of the steps, the implementation in $G^2$ rather than in $G$ adds some additional complications. In the following, we give a detailed overview over all the steps of our algorithm.

Before going into the details of the algorithm, let us note that some regimes of $\Delta$ in relation to $n$ greatly simplify the problem. If $\log n=O(\log \Delta)$, using the $O(\log n)$-time algorithm of \Cref{sec:randAlg} already yields the claimed time complexity. The problem is also simpler when $\Delta\leq \log n \cdot \poly\log \log n$, as we can then essentially simulate the preshattering algorithm of \cite{BEPS12} for $G$ on $G^2$ and combine it with our postshattering algorithm from \Cref{sec:postshattering} (for details, see \Cref{app:smallLarge}). From now on, in \Cref{sec:preshattering}, we assume to be out of those simpler regimes, i.e., we assume throughout that $\Delta=2^{o(\log n)}\cap \tilde{\Omega}(\log n)$ holds.

\subsection{Preshattering: Algorithm Overview \& Proofs}
\label{sec:preshattering}

In the shattering framework, the high level idea of the preshattering phase is that having each node try a random color a logarithmic number of times is enough to ensure w.h.p. all that is left to do is to extend a partial coloring to small connected components of size $O(\polylog n)$ and maximum degree $O(\log n)$.
But while the shattering framework with informed color trials is well established, we apply it here in an unusual setting where the nodes do not know their palette. Instead, we argue that each live node becomes colored in each iteration with constant probability (bounded away from 0). More strongly, we show that half of the colors of its palette have good probability of becoming the node's color in each round, and this holds independent of what its neighbors do (as long as the unlikely event of too many of them are activated does not happen).  Then we show that these conditions are sufficient to leave us with two disjoint subgraphs of live nodes, both of logarithmic degree, which we handle sequentially (we first execute all steps after Step~\ref{step:degreeestimation} including the postshattering phase for the one subgraph and then for the other subgraph). After conducting additional $O(\log \Delta)$ informed color trials, the uncolored vertices induce polylogarithmic size components. The rest of the coloring can then be completed in the postshattering phase.  The idea of producing two subgraphs of small degree already appeared in \cite{BEPS12}, but it is significantly easier to show that they cover all uncolored vertices if one can perform informed color trials.

Several further technical complications arise that do not occur for ordinary graph coloring: determining which of the two subgraphs the live node should join; adding Steiner nodes to make the components connected in $G$ (not just in $G^2$); and learning enough of the palette before the post-shattering phase, even when the palette might be large. 
All of these steps, however, are implementable within the $O(\log \Delta)$ time bound, with techniques of modest novelty.  The key idea for their efficient implementation is to compress the communication so that multiple messages fit in a single \CONGEST message. Color values use $\log \Delta$ bits, but we also compress node identifiers into $O(\log \Delta)$ bits, either through hashing or renumbering within a component. This allows us to speed up communication-heavy parts: $O(\log n\cdot \log \Delta)$ bits per edge can be sent in $O(\log \Delta)$ rounds.

All of the above is for dense nodes, for which we have the structure of the almost-clique decomposition to guide us. For sparse nodes, we can use simple uniformed color guessing, first with individual colors and then with parallel color guesses, to finish them off early. 

We perform the following steps. They can all be implemented in $O(\log \Delta) + \poly \log \log n$ rounds, except the postshattering phase. The cost of $\ND$ in Theorem \ref{thm:d2ColoringSublog} has its origins in the postshattering phase only.

\subsubsection{Preshattering: Algorithm Overview}

\newcommand{\phasetitle}[1]{\medskip\noindent{\textbf{\textsf{#1}}}}

\newcounter{myCounter}
\phasetitle{Almost Clique Decomposition}
\begin{enumerate}
\item \label{step:acd}
Compute the ACD exactly in $O(\log\Delta)$ rounds by hashing IDs to $O(\log\Delta)$ bits.\\
\textbf{Guarantee:} Nodes know whether they are sparse/dense. Furthermore, each dense node knows an identifier of its almost clique.
\setcounter{myCounter}{\value{enumi}}
\end{enumerate}
\phasetitle{Color Sparse Nodes}
\begin{enumerate}
\setcounter{enumi}{\value{myCounter}}
\item \label{step:sparsegivesslack} Every node (dense or sparse) tries a uniformly random color for $O(\log \Delta)$ rounds.\\
\textbf{Guarantee:} All nodes have slack proportional to their sparsity.
\item \label{step:uninformedschneiderwattenhofer} Sparse nodes try $O(\log n)$ random colors simultaneously. In total, trying $O(\log n)$ colors requires sending/receiving $O(\log \Delta\cdot \log n)$ bits to immediate neighbors, which can be sent in $O(\log \Delta)$ rounds (by packing $O(\log n)$ bits in each message). \\
\textbf{Core idea:} Each color you try has a constant probability to not be tried by anyone else nor adopted by a neighbor.\\
\custombox{
\textbf{Guarantee:} All sparse nodes are colored, w.h.p.
}
\setcounter{myCounter}{\value{enumi}}
\end{enumerate}
\medskip
\noindent Only dense (intermediate degree) nodes execute the remaining steps.

\phasetitle{Degree Reduction of Uncolored Graph}
\begin{enumerate}
\setcounter{enumi}{\value{myCounter}}
\item \label{step:degreereduction} Perform $O(\log \Delta)$ iterations of Reduce-Phase. \\
\custombox{
\textbf{Guarantee:} Uncolored nodes either have low uncolored degree (at most $\tilde{\Delta}$), or are connected to at most $\tilde{\Delta}$ other high uncolored degree nodes, where $\tilde{\Delta}=O(\log n)$.
}
\item \label{step:degreeestimation} Estimate uncolored degree with $\Theta(\log n)$ precision. \\
\textbf{Guarantee:} Uncolored nodes know whether they have low uncolored degree or not.
\setcounter{myCounter}{\value{enumi}}
\end{enumerate}

\noindent Let $U^{lo}$ and $U^{hi}$ be the sets of low and high uncolored degree vertices. All the steps afterwards first take place on $U^{lo}$, then on $U^{hi}$.

\begin{enumerate}
\setcounter{enumi}{\value{myCounter}}
\item \label{step:filteredschneiderwattenhofer} Try $\Theta(\log n)$ color proposals that arrive through parallel Reduce-Phases. \\
\textbf{Core idea:} Compressing the messages communicated in a Reduce-Phase into $O(\log \Delta)$ bits. Argue a bound of $O(\log \Delta)$ on the congestion of each edge. \\
\textbf{Guarantee:} Nodes with slack $\Omega(\log^2 n)$ become colored, w.h.p.
All remaining live nodes then have sparsity $O(\log^2 n)$ (needed for Step~\ref{step:learnlist} and \ref{step:steiner}).
\setcounter{myCounter}{\value{enumi}}
\end{enumerate}
\phasetitle{Shattering Into Small Connected Uncolored Components}
\begin{enumerate}
\setcounter{enumi}{\value{myCounter}}
\item \label{step:learnlist} \textbf{Learn your list:} Expand on the method \alg{LearnPalette} of \cite{HKM20} to have each live node learn a list of at least $d(v)+1$ available colors from its palette. If the node has sparsity $O(\log n)$, we learn the exact list using \alg{LearnPalette} as is. Otherwise, we randomly try colors not used in the almost-clique to learn enough available colors. \\
\textbf{Core idea:} The bottleneck of the method is sending $O(\log n)$ colors over a single link, i.e., $O(\log n\log\Delta)$ bits. By compressing messages this can be done in $O(\log \Delta)$ rounds.
\item \label{step:shattering} \textbf{Shattering:} Perform $O(\log \tilde{\Delta})=O(\log\log n)$ informed color tries (OneShotColoring).\\
\custombox{
\textbf{Guarantee:} Uncolored vertices induce $\poly(\tilde{\Delta})\log n=\poly\log n$ sized components in $G^2$, and uncolored vertices know a palette that exceeds their degree.
}
\item \label{step:steiner} \textbf{Add Steiner Nodes:} Add all vertices that link live nodes in different almost cliques. 
Inside each almost clique, learn all live neighbors IDs through ID-renaming, pick one intermediate node as Steiner node per pair of uncolored nodes in the almost clique.\\
\textbf{Guarantee:} $G^2[U]$ connected components are $G$-connected and of size $N=\poly\log n$.
\end{enumerate}
\phasetitle{Postshattering:}
 Before the process, uncolored dense nodes $U$ form small connected components and each node has a palette of size that exceeds its degree. Further, with the Steiner nodes connected components of $G^2[U]$ are $G$-connected and have $N=\poly\log n$ size. 
 This is enough to apply \Cref{lem:postShattering} in \Cref{sec:postshattering} and list color the remaining components in $\ND = 2^{O(\sqrt{\log\log n})}$ rounds.

\subsubsection{Step \ref{step:acd}: Implementing the ACD} 
\label{ssec:hashing}
We start by computing the almost clique decomposition, which relies on computing two predicates,  \textsc{Buddies} or \textsc{Popular} (see Appendix~\ref{app:acd} for definitions and details). To implement \textsc{Buddies} or \textsc{Popular}, the nodes need to inform their d2-neighbors that they are in the set $S$ and forward their $S_v$ sets to their intermediate neighbors, where $\E[|S|] = c_{10}\log n$. Instead of using original node IDs in this process, we have each node $v$ pick a random string $h_v$ in the range $\eta\cdot \Delta^4$ for a sufficiently large constant $\eta>0$, to use instead. Since nodes forward hashed values, $h_w$, that fit in $O(\log \Delta)$ bits, forwarding the sets $S_v$ runs in $O(\log \Delta)$ rounds, since it just involves forwarding $O(\log n \cdot \log \Delta)$ bits.

The hashes may collide, which results in an undercount of the size of each $S_v$ set. However, by Lemma~\ref{lem:node-hashing}, this underestimate is at most additive $\frac{\ln n}{\ln\ln n} = o(\log n)$ , which disappears into the concentration bounds for $|S_{uv}|$, since $|S_{uv}| = \Omega(\log n)$, w.h.p.

\begin{lemma}[Hashing Node IDs]
For each node $v$, the number of different hash values in the d2-neighborhood is at least $|\{h_u : u \in N_{G^2}(v)\}| \geq |N_{G^2}(v)| - \frac{\ln n}{\ln\ln n}$, w.h.p.
\label{lem:node-hashing}
\end{lemma}
\begin{proof}
   Assume that $N_{G^2}(v)=\{u_1,u_2,\dots,u_\ell\}$, where $\ell=|N_{G^2}(v)|$. We define an indicator random variable $X_i$ for each node $u_i$, $i\in\{1,\dots,\ell\}$ such that $X_i=1$ iff $h_{u_i}=h_{u_j}$ for some $j<i$. Note that independently of the hash values of $h_{u_j}$ for $j<i$, we have $\Pr[X_i=1]\leq 1/(\eta\Delta^2)$. The variables $X_i$ for are thus dominated by a set of independent indicator variables $Y_i$ such that $\Pr[Y_i=1]=1/(\eta\Delta^2)$. The number of hash value collisions can then be upper bounded by $Y:=Y_1+\cdots+Y_\ell$ and by applying a standard Chernoff bound, we have $\Pr\big(Y > \frac{\ln n}{\ln\ln n}\big)<n^{-\Theta(\eta)}$.
\end{proof}

\subsubsection{Step \ref{step:sparsegivesslack}--\ref{step:uninformedschneiderwattenhofer}: Coloring Sparse Nodes}

The next two steps of the algorithm color all sparse nodes w.h.p., leaving us with only the dense nodes to deal with later on. In addition, the first uninformed color try guarantees that all nodes of sparsity $\Omega(\log n)$ have slack at least linear in their sparsity, per Proposition~\ref{P:sparsity}, which will be useful later on when we focus on coloring dense nodes of sparsity $\Omega(\log n)$.

Consider the sparse nodes $V^\ast$ of the almost clique decomposition (Definition~\ref{D:acd}(1)). In step \ref{step:sparsegivesslack}, all nodes do $O(\log \Delta)$ uninformed color tries. We show this is enough to ensure w.h.p. that after this step, each sparse node has $O(\log n)$ sparse nodes in its d2-neighbourhood.

\begin{lemma}[Step \ref{step:sparsegivesslack}]
Let $S_2$ be the uncolored sparse vertices after Step \ref{step:sparsegivesslack}. $G^2[S_2]$ has maximum degree $O(\log n)$, w.h.p..
\end{lemma}
\begin{proof}
We prove the stronger claim that at the end of Step \ref{step:sparsegivesslack}, each node $v\in V$ has at most $O(\log n)$ uncolored sparse neighbors in its d2-neighbourhood.

After one round of uninformed color tries, Proposition~\ref{P:sparsity} tells us that nodes of sparsity $\zeta \in \Omega(\log n)$ have slack $\Omega(\zeta)$ w.h.p.. As $\Delta = \tilde\Omega(\log n)$, $\Delta^2 \in \tilde\Omega(\log^2 n) \subseteq \Omega(\log n)$, so the sparse nodes $V_\ast$ identified in Step~\ref{step:acd}, of sparsity $\frac{\epsilon^2}{4}\Delta^2$ by definition, all have $\frac{\epsilon^2}{16e^3}\Delta^2$ slack w.h.p.. Let $p:= \frac{\epsilon^2}{16e^3}$ in the context of this proof.

Consider the d2-neighbourhood of an arbitrary vertex $v$ at an arbitrary iteration of this step. Let $S^0_v$ be the set of sparse nodes in $v$'s neighbourhood and $S^1_v$ the same set after an additional round of uninformed color tries. $\E[|S^1_v|] \leq p \E[|S^0_v|]$. Moreover by Chernoff (Proposition~\ref{P:chernoff}), $\Pr[|S^1_v| \geq \frac{1+p} 2 |S^0_v|] \leq \exp \left( -\frac{p(1-p)^2} {12} |S^0_v| \right)$.

Therefore in any d2-neighbourhood of a node $v$ that contains more than $\frac {24} {p(1-p)^2} \log n$ sparse nodes at some iteration of this step, it holds w.h.p. $\geq 1 - n^{-2}$ that a fraction at least $\frac p 2$ of these sparse nodes gets colored in the next round of uninformed color tries. Therefore by union bound over the $n$ d2-neighbourhoods and all the iterations of the algorithm, and since every d2-neighbourhood contains at most $\Delta^2$ in general, it holds w.h.p. that after $\log_{2/p} \Delta^2$ rounds of uninformed color tries every d2-neighbourhood contains less than $\frac {24} {p(1-p)^2} \log n$ sparse nodes.
\end{proof}

The fact that each remaining live sparse node does not have too many other live sparse nodes in their neighbourhood makes trying colors in batch a viable strategy, used in step \ref{step:uninformedschneiderwattenhofer}.

\begin{lemma}[Step \ref{step:uninformedschneiderwattenhofer}]
After Step \ref{step:uninformedschneiderwattenhofer} each sparse node is colored, w.h.p. and the step can be implemented in $O(\log \Delta)$ rounds.
\end{lemma}
\begin{proof}

Let $S\subseteq V_\ast$ denote the set of live sparse nodes.
Let $p=\epsilon^2 /(16 e^3)$. Each node of $S$ has slack $s \ge p \Delta^2$ (Definition~\ref{D:acd}(1) and Proposition~\ref{P:sparsity}), and slack never goes down.
Let $d = O(\log n)$ be the maximum degree of the graph $G^2[S]$ induced by the live sparse nodes, $m = s/(2d)$, and $q = (6/p) \log n$. 
We run the \alg{MultiTrial} procedure of~\cite{SW10}.

Suppose first that $m \ge q$. Then each node of $S$ tries $q$ colors uniformly at random. This amounts to $q \log \Delta = O(\log \Delta \cdot \log n)$ bits, which can be transmitted in $O(\log \Delta)$ rounds.
The total number of colors chosen by the at most $d$ sparse d2-neighbors of a live sparse node is $d\cdot q \le d\cdot m = s/2$. Hence, each of the $q$ colors has probability at least $(s/2)/\Delta^2 = p/2$ of succeeding, and the probability of some color succeeding is at least $1 - (1-p/2)^q \ge 1 - e^{3\ln n} = 1-n^{-3}$.

When $m < q$, we repeat the procedure $q/m\in O((\log n)/\Delta^2) \subseteq O(1)$ times, trying $m$ colors each time, for the same performance bound.
\end{proof}

\subsubsection{Step \ref{step:degreereduction}--\ref{step:filteredschneiderwattenhofer}: Degree Reduction}
The next steps have for goal to reduce our coloring problem to coloring problems on graphs of small degree. Consider a constant $C$, and define $U^{lo}$ to be the set of uncolored nodes of uncolored degree at most $C \log n$, and $U^{hi}$ to be the other uncolored nodes. Step~\ref{step:degreereduction} ensures that while there might be nodes of high uncolored degree, those high uncolored degree node can not have many d2-neighbors also of high uncolored degree. Step~\ref{step:degreeestimation} is then there to have the nodes learn if they have high or low uncolored degree. The partition of the nodes into low and high degree nodes does not have to be perfect, we can tolerate an $O(\log n)$ gray zone of degrees in which nodes of those degrees may end up in either set. Steps~\ref{step:learnlist} to~\ref{step:steiner} and the postshattering phase are then run on $U^{lo}$, then $U^{hi}$.

\paragraph{Step~\ref{step:degreereduction}: Splitting the Uncolored Nodes into Two Subgraphs of Small Maximum Degree}

\begin{lemma}[Step~\ref{step:degreereduction}]
\label{lem:degreeReduction}
There is an universal constant $c$ such that after Step~\ref{step:degreereduction}, every live node either has uncolored d2-degree at most $c \log n$, or it has at most $c \log n$ uncolored nodes of d2-degree greater than $c \log n$ in its d2-neighborhood, w.h.p.
\label{lem:stepdegreereduction}
\end{lemma}

To prove this, we adapt a result of~\cite{BEPS12} that shows that a few rounds of good random color tries on a graph guarantees that the remaining uncolored vertices can be partitioned into two sets, one of low uncolored degree, the other inducing a graph of low degree. This is detailed in the statement of Lemma~\ref{lem:shatterdegreereduction}.
\begin{lemma}[Adaptation of Lemma 5.4 of~\cite{BEPS12}]
    Let $p,f_1,f_2$ be constants such that $0 \leq f_1 < f_2 \leq 1$ and $p>0$. Let us have access to an algorithm $\calA$ to randomly try colors such that an iteration of $\calA$ is such that for every node of uncolored degree at least $\Omega(\log n)$:
    \begin{itemize}
        \item at most a constant fraction $f_1 |\psi_v|$ of its palette is tried by its neighbors in this iteration, w.h.p.
        \item the colors of $v$'s palette that $v$ tries with probability $\geq \frac p { |\psi_v|}$ represent a constant fraction $f_2 |\psi_v|$ of its palette.
    \end{itemize}
    Then $O(\log \Delta)$ iterations of $\calA$ guarantee w.h.p. that the remaining uncolored vertices $U$ can be partitioned into two sets $U^{lo}$ and $U^{hi}$ such that:
    \begin{itemize}
        \item the vertices in $U^{lo}$ have maximum uncolored degree $O(\log n)$,
        \item the subgraph induced by $U^{hi}$ has maximum degree $O(\log n)$.
    \end{itemize}
    \label{lem:shatterdegreereduction}
\end{lemma}
\begin{proof}
    Let $c$ be a constant and $U^{hi} = \{v \in U, deg_U(v) > c \log n\}$ be the uncolored vertices of large uncolored degree.
    
    Consider a vertex $v\in U^{hi}$ of high degree within $U^{hi}$, i.e., $deg_{U^{hi}(v)} > c \log n$. We show that through the iterations of $\calA$ it will either exit $U^{hi}$ by getting colored, or stay in $U^{hi}$ but have a low degree in the subgraph induced by $U^{hi}$, w.h.p.
    
    We show that in a round of $\calA$, a node $v\in U^{hi}$ either gets colored or has its $U^{hi}$ degree decrease geometrically. Assume that all nodes in $U^{hi}$ are such that a fraction at most $f_1$ of their palette is tried by their neighbors, which we know by assumption holds w.h.p. Let us consider a node $v\in U^{hi}$, and its $U^{hi}$ neighbors $N_{U^{hi}}(v)$ in increasing ID-order $u_1, \ldots, u_{deg_{U^{hi}(v)}}$.
    
    Let $X_i$ be the event that $u_i$ gets colored, and $Y_i$ represent the information about the first $i$ neighbors of $v$. We want to show that $X = \sum_{1 \leq i \leq deg_{U^{hi}}(v)} X_i$, the number of neighbors of $v$ that get colored in a round, is larger than a constant fraction of $deg_{U^{hi}}(v)$ with high probability, meaning that the uncolored degree of $v$ decreases geometrically.
    
    We condition on the event that for every node of uncolored degree $\Omega(\log n)$, its neighbors are trying at most $f_1|\psi_v|$ colors, which holds with high probability by assumption.
    For any behavior of its neighbors, conditioned on this event, there are at least $(f_2 - f_1)|\psi_v|$ colors that $v$ is trying with probability at least $p/|\psi_v|$. Therefore $\Pr[X_i=1 | Y_{i-1}] \geq p (f_2 - f_1)$ for every $i$.
    
    This immediately implies that $\E[X]\geq p (f_2 - f_1) deg_{U^{hi}}(v)$. By a concentration argument (Corollary A.5 in~\cite{BEPS12}), we have that for $c$ large enough, there exists a constant $\gamma$ such that for all $U^{hi}$ nodes of high $U^{hi}$ $c \log n$, the probability that its $U^{hi}$ degree decreases by at least a constant fraction $\gamma$ satisfies $\Pr[X > \gamma \cdot deg_{U^{hi}}(v)] \geq 1 - n^3$. Therefore it holds for all nodes of $U^{hi}$ in a round of $\calA$ with probability $\geq 1-n^2$, and for $O(\log \Delta)$ iterations with high probability. Since this degree is at most $\Delta^2$ to start with, and this geometric decay works as long as the $U^{hi}$ degree is at least $c \log n$, $O(\log \Delta)$ rounds of $\calA$ suffice to ensure that all $U^{hi}$ nodes have less than $c \log n$ $U^{hi}$ neighbors with high probability.  
\end{proof}

\begin{proof}[Proof of \Cref{lem:degreeReduction}]Consider the \alg{Reduce-Phase} algorithm. In an iteration of this algorithm, at the very start each node randomly decides to try a color in this iteration with probability $1/8$, and to stay silent otherwise. This ensures that for nodes of uncolored degree $\Omega(\log n)$, a fraction at most $1/4$ of its palette is tried by its neighbors with high probability. Lemma~\ref{L:phi-ophi} and~\ref{L:ophi} ensure every node $v$ has a fraction at least $2/3$ of its palette $\psi_v$ that it tries with probability at least $p/|\psi_v|$, where $p$ is a universal constant.

We can therefore apply Lemma~\ref{lem:shatterdegreereduction} with the \alg{Reduce-Phase} algorithm and obtain the desired degree reduction.
\end{proof}

\paragraph{Step~\ref{step:degreeestimation}: Estimating Your Uncolored Degree}

\begin{lemma}[Step~\ref{step:degreeestimation}] After Step~\ref{step:degreeestimation}, the live nodes are partitioned into two sets $U^{lo}$ and $U^{hi}$ such that every live node of uncolored degree less than $c \log n$ has joined $U^{lo}$, every live node of uncolored degree greater than $2c \log n$ has joined $U^{hi}$, while the other live nodes may have joined either, w.h.p.
\end{lemma}

After Step~\ref{step:degreereduction}, it is guaranteed the uncolored nodes can be split into small and high degree nodes with the additional constraint that the high degree nodes are not connected to a lot of other high degree nodes. It does not, however, give this decomposition out of the box, as the nodes may not know their uncolored degree. Let $c > 10$ be a constant such that our application of Lemma~\ref{lem:stepdegreereduction} in the previous step guarantees that nodes of uncolored degree $\geq c\log n$ are connected to at most $c \log n$ other high degree nodes. We have the nodes estimate their degree such that all the live nodes of uncolored degree less than $c \log n$ join the set $U^{lo}$, while all the live nodes of uncolored degree at least $2 c \log n$ join $U^{hi}$. Nodes of uncolored degree between $c \log n$ and $2 c \log n$ may join either set.

\begin{lemma}
    Let $G$ be a partially colored graph where only dense nodes of the ACD are still uncolored. There is an $O(\log \Delta)$ round algorithm that allows each node to know whether it has uncolored degree $O(\log n)$ in $G^2$, w.h.p.
\label{lem:estimatedegree}
\end{lemma}
\begin{proof}
 First, each node counts how many uncolored direct neighbours it has. If that number is greater than $c \log n$, it informs all its direct neighbours that they have a high uncolored d2-degree. We now only need to deal with live nodes in the direct neighbourhood of nodes that have less than $c \log n$ uncolored direct neighbours. 
 Each uncolored vertex picks a random number in $[\Delta^4]$ and broadcasts these for two hops. This broadcast takes at most $O(\log \Delta)$ rounds as each vertex has at most $O(\log n)$ uncolored neighbors, and so every vertex has to forward at most $O(\log \Delta\cdot \log n)$ bits in total, which can be done in $O(\log \Delta)$ rounds. 
 An uncolored vertex joins the set $U^{hi}$ of high degree vertices if it receives more than $c \cdot \log n$
 distinct values, otherwise it joins the set $U^{lo}$. 
 A vertex which receives $c\log n$ distinct values clearly has more than $c\log n$ uncolored $d2$-neighbors. We next prove that a vertex that receives less than $c\log n$ distinct values has at most $2c \log n$ uncolored neighbors, or equivalently, that any vertex with at least $l=2c\log n$ uncolored neighbors receives at least $c\log n$ values. 
 Fix a vertex $u$ and let $v_1,\ldots, v_l$ its uncolored neighbors in an arbitrary order. We expose the randomness one after the other and let $X_i$ be the random variable that equals $1$ if $v_i$ hashes to a value that appears in the hashes of the vertices $v_1,\ldots,v_{i-1}$, and $X_i$ equals zero otherwise. We have that 
 \begin{align*}
 \Pr[X_i=1 \mid v_1,\ldots,v_{i-1} \text{ are hashed}]\leq (i-1)/(\Delta^4)\leq (2c\log n)/\Delta^4=:p. 
 \end{align*}
 Thus, we obtain that the number of vertices in a collision is in expectation upper bounded by $\E[\sum_{i=1}^lX_i] \leq p\cdot l$, and with a Chernoff bound the number of collisions is less than $c\log n$, w.h.p.. Thus, w.h.p. $u$ receives at least $2c\log n-\#collision=c\log n$ distinct values.
\end{proof}

\paragraph{Step \ref{step:filteredschneiderwattenhofer}: Parallelizing Reduce-Phase} 

\begin{lemma}[Step \ref{step:filteredschneiderwattenhofer}]
After Step~\ref{step:filteredschneiderwattenhofer}, all the live nodes of $U$ of slack $\Omega(\log^2 n)$ have been colored, thus leaving only live nodes of slack $O(\log^2 n)$, and a fortiori sparsity $O(\log^2 n)$.
\end{lemma}

This step takes $O(\log \Delta)$ rounds. The two key ideas are to try multiple colors at the same time, as in Step~\ref{step:uninformedschneiderwattenhofer}, which uses the \alg{MultiTrial} procedure of Schneider and Wattenhofer~\cite{SW10}, and to use a smaller ID space when using \alg{Reduce-Phase}. What allows us to use an ID space of size $O(\poly \Delta)$ is that in \alg{Reduce-Phase} nodes only communicate with other nodes of their extended connected component, and the small radius of those extended connected components.
\begin{proof}
The almost clique decomposition guarantees that each node knows the ID of the extended component it belongs to. Having each node inform its direct neighbors of its component ID only takes a single round, after which every node knows the component ID of all of its neighbours.

Then, in each extended component of the almost-clique decomposition, we give each node an ID in $[2\Delta^2]$ that uniquely identifies them inside their extended component. Let us call those numbers \emph{local IDs}. This is done with a BFS traversal of the extended component, which we can do in $O(1)$ rounds since any two nodes in an almost-clique are within four hops.

Consider the algorithm \alg{Reduce-Phase}. In it, queries for help and answers are only sent from nodes inside an extended component to other nodes of the extended component. In addition, the messages either transmit colors or IDs of nodes of the extended connected component, and when asking whether two nodes are connected, it only does so for two nodes of the same extended component. With this observation, we can replace the IDs with local IDs in \alg{Reduce-Phase}. With this trick, an iteration of \alg{Reduce-Phase} only needs to send $O(\log \Delta)$ bits. Thus, we can perform $\Omega(\log_\Delta n)$ phases in parallel.

As in Step~\ref{step:uninformedschneiderwattenhofer}, there is a universal constant $c$ such that each uncolored node in this step has at most $c \log n$ uncolored d2-neighbors. So if each uncolored node tries at most $d \log n$ colors for some constant $d$, every node sees at most $2cd \log^2 n$ colors being tried by its neighbors. As a node $v$ doing an iteration of \alg{Reduce-Phase} is trying with decent probability $p/|\psi_v|$ at least a third of its palette $\psi_v$ by Lemma~\ref{L:phi-ophi} and~\ref{L:ophi}, where $p$ is an absolute constant, a node of slack at least $6cd\log^2 n$ succeeds with constant probability with each color it tries. By having each node try $\Theta(\log n)$ colors, using $\Theta(\log \Delta)$ iterations of \alg{Reduce-Phase} executed in parallel in $\Theta(\log \Delta)$ \CONGEST rounds, every node of slack more than $6cd\log^2 n$ gets colors w.h.p. and the remaining live nodes all have slack and sparsity $O(\log^2 n)$.
\end{proof}

\subsubsection{Step~\ref{step:learnlist}--\ref{step:steiner}: Shattering into Small Connected Components}

The last three steps of the preshattering phase end our preparation for the postshattering phase, described later in Section~\ref{sec:postshattering}. One of the main goals of the preshattering phase is to reach a situation where the uncolored part of the considered graph consists of small connected components, which can thus be treated independently (since they are disconnected from each other) and fast (since they are much smaller than the original graph). Breaking the problem into small connected components is done in step \ref{step:shattering}. As this step uses informed color tries, a bit of preparation is needed beforehand: step~\ref{step:learnlist} serves to have each remaining dense node learn more colors from its palette than its uncolored degree. Finally, step \ref{step:steiner} serves to recruit nodes as relays to help with communication in the postshattering phase. The intuition behind it is that the vertices of a connected component of $G^2$ are not necessarily connected in $G$, so we consider additional vertices to connect and guarantee a good bandwidth between each pair of d2-neighbors.

\paragraph{Step~\ref{step:learnlist}: Learn Your List}

The method \alg{LearnPalette} of \cite{HKM20} allows each live node to learn its available palette, assuming it has palette size $O(\log n)$. We extend it to apply to the case when it has degree $O(\log n)$, but the palette size could be larger.
We will run it (and the procedures that follow) separately on $U^{lo}$ and $U^{hi}$.

We first apply the first phase of \alg{LearnPalette} of \cite{HKM20} essentially unchanged. Each live node $v$ randomly selects an $H$-neighbor \emph{helper} $z_v^i$ for each of the $\Delta$ blocks of $\Delta$ colors from the color space $[\Delta^2]$. Then, nearly all d2-neighbors of $v$ forward their color to the appropriate helper. In our context, it means that all nodes within the extended component register its color. 
This phase actually consists of only four steps, where each link may need to forward $\Theta(\log n \cdot \log \Delta)$ bits: $\Theta(\log n)$ different color values. Now, the gaps in the blocks of the helpers are of small size: its union consists of two parts: $\psi_v$, the true palette for $v$, and $\nu_v$, the colors of nodes outside the extended component. By assumption, $|\psi_v|=O(\log n)$, while by sparsity and Lemma \ref{L:h-degree}(1), $|\nu_v| \le (2\zeta+1)/\epsilon = O(\log n)$. 

\begin{lemma}
The first phase of \alg{LearnPalette} of \cite{HKM20} runs in $O(\log \Delta) + poly(\log\log n)$ rounds. It achieves the following: for each live node $v$ and each $i=1,2\ldots, \Delta$, there is a \emph{helper} node $z^i_v$ that stores a set $T_v^i$ such that:
\begin{itemize}
    \item The paths $P_{i,v} = [z_v^i,v]$ are edge-disjoint.
    \item The set $T_v = \cup_i T_v^i$ contains the palette $\psi(v)$ ($T \supseteq \psi_v$) and $|T_v \setminus \psi_v| = O(\zeta_v)$.
\end{itemize}
\end{lemma}

If $\zeta_v = O(\log n)$, then $v$ learns the exact palette in $O(\log n)$ rounds, which can be compressed into $O(\log \Delta)$ rounds. 

\begin{lemma}[Step \ref{step:learnlist}, learn your list]
    Let $G$ be a partially colored graph where uncolored nodes have uncolored degree at most $D=O(\log n)$ in $G^2$ and slack $O(\log^2 n)$. There is an $O(\log \Delta + \polyloglog n)$-round algorithm that allows each node to learn a palette of size $\min(|\psi_v|, D+1)$.
\label{lem:learnlistsublog}
\end{lemma}
\begin{proof}
By querying the helpers, we compute $|T_v|$. If $|T_v| = O(\log n)$, 
we can then use the second phase of \alg{LearnPalette} unchanged. All helpers can then forward their lists $T_v^i$ to $v$ in $O(\log \Delta)$ rounds, by compressing $\log_\Delta n$ colors in a single message. Next, $v$ sends its combined list $T_v = \cup_i T_v^i$ to its immediate neighbors, who report the colors in $T_v$ that are already used. Now, $v$ has learned its true palette $\psi_v$ in $O(\log \Delta)$ rounds.

If $|T_v| = \Omega(\log n)$, we need a modified approach.
Since the slack is known to be $O(\log^2 n)$, sparsity is also $O(\log^2 n)$.
We learn $D + 1$ palette colors as follows. 
Each helper $z_v^i$ sends to $v$ randomly chosen $\log_\Delta n$ colors from the possibly available colors $T_v^i$. $v$ then picks uniformly (by weighing the choices according to $|T_v^i|/|T_v|$) $\log_\Delta n$ colors and tries them (forwards to immediate neighbors and learns which ones were already used).
Each color query has probability $|\psi_v|/(|T_v| = \Omega(1)$ of being in the palette. Thus, by Chernoff, it suffices to query $O(D)+O(\log n) = O(\log n)$ colors to learn $D+1$ colors from the palette. We are then ready to solve a deg+1-list coloring instance.
\end{proof}

The combined color gaps due to colors of nodes outside $\hat{C}$ is $|\nu_v| \le (2\zeta+1)/\epsilon$ such colors, by Lemma \ref{L:h-degree}(1). It suffices for us learn of only $|D+1|$ colors from the palette, where $D = O(\log n)$ is the live degree of $v$.

\paragraph{Step~\ref{step:shattering}: Shattering Into Small Connected Components}
\begin{lemma}[Step~\ref{step:shattering}] After step~\ref{step:shattering}, the subgraph $G^2[U]$ induced by the subset $U\in \{U^{lo},U^{hi}\}$ of live nodes currently being considered has connected components of size most $\polylog n$, w.h.p.
\label{lem:stepshattering}
\end{lemma}

Lemma~\ref{lem:stepshattering} follows a result of~\cite{BEPS12}, adapted to our setting (Lemma~\ref{lem:shattersmallcomponent}), that shows that a few rounds of good random color tries on a graph whose uncolored nodes induce a subgraph of low maximum degree guarantee that, after the random color tries, the connected components of the subgraph induced by the uncolored nodes are small.

\begin{lemma}[Adaptation of Lemma~5.3 of~\cite{BEPS12}]
    Let $G$ be a partially colored graph, $U$ a subset of the uncolored nodes and $\Deltahat$ be the maximum degree in $G^2[U]$, the subgraph of $G^2$ induced by $U$.
    Let us have access to an $O(1)$ round algorithm $\calA$ such that:
    \begin{itemize}
        \item each uncolored node of $U$ gets colored with constant probability at least $p_{succ}$ in all but at most $O(\log \Delta)$ iterations of $\calA$,
        \item for all uncolored nodes $v\in U$, let $E_v$ be the event that $v$ gets
        colored by the algorithm. There exists $d$ such that $E_v$ is independent
        of all events $\{E_u: d_{G^2}(u,v) \geq d\}$.
    \end{itemize}
    Then $O(\log \Deltahat)$ iterations of $\calA$ suffice to guarantee that
    all uncolored components in $G^2$ have less than
    $O(\log_{\Deltahat}(n) \Deltahat^{d-1})$ nodes with high probability.
    \label{lem:shattersmallcomponent}
\end{lemma}
\begin{proof}
    Let us consider a distance-$d$ set of size $t=c\log_{\Deltahat} n$, meaning a set of uncolored nodes that are at distance at least $d$ from each other that form a tree in the uncolored part of $G^{2d}$. There are at most $4^t \cdot n \cdot {\Deltahat}^{d(t-1)} \leq n^{2c/\log \Deltahat + 1 + cd} \leq n^{1+c(2+d)}$ such sets.
    
    Let an iteration of $\calA$ be \emph{good} for $v$ if it guarantees that $v$ gets colored with probability at least $p_{succ}$, and \emph{bad} for $v$ otherwise. Let $p^{i,v}_{fail}$ be the probability that a node $v$ stays uncolored in iteration $i$ of $\calA$, and $p_{fail} = 1 - p_{succ}$ be maximum probability that a node stays uncolored in a good iteration for $v$. The probability that a given distance-$d$ set $T$ remains completely uncolored after one round of algorithm $\calA$ is at most $\prod_{v \in T}p^{1,v}_{fail}$. Conditioned on the event that the set was not colored in the first $i-1$ rounds, the probability that it stays uncolored after round $i$ is at most $\prod_{v \in T}p^{i,v}_{fail}$. Thus the probability that the set remains uncolored after executing $\calA$ for $r$ rounds is at most $\prod_{i \in [r]}\prod_{v \in T}p^{i,v}_{fail} = \prod_{v \in T}\prod_{i \in [r]}p^{i,v}_{fail}$. Let $r_{bad} \in O(\log \Delta)$ be an upper bound on the maximum number of bad iterations a node has. Running $\calA$ for $r+r_{bad}$ iterations ensures that a distance-$d$ set $T$ remains uncolored with probability at most $\prod_{v \in T}\prod_{i \in [r+r_{bad}]}p^{i,v}_{fail} \leq \prod_{v \in T} (p_{fail})^{r} = (p_{fail})^{r \cdot t}$.
    
    By union bound, the probability that there exists one distance-$d$ set that remains uncolored after $r+r_{bad}$ rounds of $\calA$ is at most $(p_{fail})^{r \cdot t} \cdot n^{1+c(2+d)} = n^{-r \frac{c\log(1/p_{fail})}{\log \hat\Delta}+(1+c(2+d))}$.
    
    Thus for any constant $b$, running algorithm $\calA$ for $r+r_{bad}$ rounds with $r = \log \Deltahat \cdot \frac {1}{c \log(1/p_{fail})} (1+c(2+d) + b) \in O(\log \Deltahat)$ ensures that no distance-$d$ set remains completely uncolored with probability $\geq 1 - n^{-b}$. As any connected component of size $t \Deltahat^{d-1}$ in $G^2$ must contain such a distance-$d$ set, not such connected component exists.
\end{proof}

\begin{proof}[Proof of \Cref{lem:stepshattering}]
    In Step~\ref{step:learnlist}, the uncolored nodes we are currently considering learned their palette. A node that knows its palette can do \emph{informed} color tries, that succeed with constant probability. Consider the following 3-round algorithm $\calA$. In the first round, each live node tries a color in its palette. In the second round, nodes that just got colored inform their direct neighbors of the success of their trial and their new color. Colors received in this round are pipelined for future 1-hop broadcast. In the third round, each node that has colors pipelined inform its direct neighbors of as many of them as it can in a single round ($O(\log n / \log \Delta)$).
    
    Let us consider the viewpoint of a single node $v$. After a round in which its direct neighbors were able to transmit all the colors they had to transmit to $v$, $v$ knows its palette perfectly, as it was informed of the evolution of all its previously live d2-neighbours. Occasionally, a direct neighbor of $v$ may have more than $O(\log n/\log \Delta)$ colors pipelined, and after an iteration where this happens it may be that $v$ is not able to color itself with constant probability because it does not know its full palette. However, there are at most $O(\log \Delta)$ iterations in which this can happen, since each node initially has at most $O(\log n)$ uncolored d2-neighbours, hence at most $O(\log n)$ colors to receive over the iterations of $\calA$, and a node receives $\Omega(\log n / \log \Delta)$ colors in an iteration where the direct neighbors were not able to transmit all their colors.
    
    Applying \Cref{lem:shattersmallcomponent} with this algorithm $\calA$ as subroutine, $\Deltahat \in O(\log n)$ and $d=3$, we obtain that $O(\log \Delta + \log \log n)$ iterations of $\calA$ suffice to guarantee that the subgraph of $G^2[U]$ induced by the remaining uncolored nodes of $U$ has connected components of size at most $\widetilde O(\log^3 n)$ w.h.p 
\end{proof}

\paragraph{Step~\ref{step:steiner}: Adding Steiner Nodes}
Let $U$ ($=U^{lo}$ or $=U^{hi}$) be the uncolored nodes that we consider in this step. In the proof of the following statement we use that connected components of $G^2[U]$ have size $\polylog n$ (by Step~\ref{step:shattering}), and that live nodes have sparsity $O(\log^2 n)$ (by Step~\ref{step:filteredschneiderwattenhofer}). 
\begin{lemma}[Steiner nodes]
There is a $O(\log\log n)$-round algorithm to select a subset $S\subseteq V$ of vertices such that:
\begin{enumerate}
\item For any $u\in U$ and any of its $d2$-neighbor $u'\in U$  there exists some $s\in S$ such that $s$ is neighbor of $u$ and $u'$, and
\item Let $K=G[U\cup S]\setminus E(G[S])$ the subgraph of $G$ induced by $U\cup S$ without edges between vertices in $S$. The connected components of $K$ have size at most $N=\poly\log n$
\end{enumerate}
\label{lem:steinerNodes}
\end{lemma}
\begin{proof}
   By Lemma~\ref{L:h-degree}(1), each live node has $O(\zeta) = O(\log^2 n)$ 2-paths to nodes outside its almost-clique $C$. We add to $S$ each such node that connects live nodes in different almost-cliques. 
   
   To identify such nodes within $C$, we first renumber the nodes of $C$ with BFS and aggregation, to use $O(\log\Delta)$ bits. Each intermediate node can then inform its live neighbors of its other live neighbors using $O(\log n \log \Delta)$ bits, thus $O(\log \Delta)$ rounds. Each live node then unilaterally chooses a single intermediate node to each of its live d2-neighbors and adds to $S$. 
   All in all, the addition of the connecting nodes increases the size of each connected component by an $O(\log^2 n)$-factor as each live node chooses only one intermediate node to connect to each of its $O(\log n)$ live d2-neighbors, and at most $O(\log^2 n)$ intermediate nodes to connect to live d2-neighbors outside its ACD component.
\end{proof}

\subsection{Postshattering: Algorithm Overview \& Proofs}
\label{sec:postshattering}
The high level idea is to compute a network decomposition $\mathcal{D}$ on each connected component of uncolored vertices to split the components into small diameter clusters. Afterwards, we use the deterministic $(deg+1)$-list coloring algorithm from \cite{BKM19} on each cluster (iterating through the clusters in an order that is given by $\mathcal{D}$).  To obtain an efficient algorithm, we need a network decomposition with two features: a) it handles distance-2 relations, and b) it handles large node identifiers (in comparison with the component sizes). The latter is not handled by the new $\poly\log n$ result of Rozho\v n and Ghaffari \cite{RG19}. Hence, we cannot currently reduce the dependence on $n$ in the time complexity to $\poly \log\log n$. Instead, the construction of Portmann and Ghaffari \cite{Portmann19} handles both of these features. The downside is the resulting time complexity of $2^{O(\sqrt{\log\log n})}$. Further, the runtime of the  list-coloring algorithm in \cite{BKM19} depends on the size of the colorspace, and we equip our algorithm with methods to reduce the colorspace before we apply \cite{BKM19}. 

\subparagraph*{Preconditions:}
We are given an $n$-vertex graph $G$ with maximum degree $\Delta$ and a partial $d2$-coloring $\phi:V\rightarrow [\Delta^2]\cup \{\bot\}$. Let $U =\{\phi^{-1}(\bot)\}\subseteq V$ be the uncolored vertices. Further, we are given a subset $S\subseteq V$ and $\Deltahat=O(\log n)$ such that:
\begin{itemize}
    \item Each node $u\in U$ has at most $\Deltahat$ $d2$-neighbors in $U$. \\
    This immediately implies that each node in $V$ has at most $\Deltahat$ $U$-neighbors in $G$.
    \item $d2$-connected components of $G^2[U]$ have size $\poly(\Deltahat)\log n=\poly\log n$.
    \item For any $u\in U$ and any of its $d2$-neighbor $u'\in U$  there exists some $s\in S$ such that $s$ is neighbor of $u$ and $u'$. 
    \item Let $K=G[U\cup S]\setminus E(G[S])$ the subgraph of $G$ induced by $U\cup S$ without edges between vertices in $S$. The connected components of $K$ have size at most $N=\poly\log n$. And,
    \item Each vertex $u\in U$ is equipped with a list $L_u$ of colors that are not used in its $d2$-neighborhood. The size of $|L_u|\leq L\leq O(\log n) \leq N$.
    \end{itemize}
\begin{lemma}[Postshattering]
\label{lem:postShattering}
There is a deterministic \CONGEST algorithm on communication network $G$ that, under the above assumptions, list colors the nodes in $U$ such that  $d2$-neighbors pick distinct colors. The runtime of the algorithm is $2^{O(\sqrt{\log\log n})}$ rounds.
\end{lemma}
The preconditions for \Cref{lem:postShattering} are satisfied after the last preshattering step (due to \Cref{lem:shattersmallcomponent,lem:steinerNodes,lem:learnlistsublog}).
\Cref{lem:postShattering} uses two subroutines from previous work. First, a network decomposition algorithm that works for $G^k$ and does not rely on a small IDspace.

\begin{definition}[Network Decomposition, $x$-CONGEST-routable \cite{awerbuch89}]
\label{def:decomposition}
  A weak \emph{$\big(d(n),c(n)\big)$-network-decomposition} of an
  $n$-node graph $G=(V,E)$ is a partition of $V$ into clusters such
  that each cluster has weak  diameter at most $d(n)$ and the
  cluster graph is properly colored with colors $1,\dots,c(n)$. 
	If the decomposition is equipped with a routing backbone such that one can simulate one round of communication within clusters of $G^k$ in $k\cdot x$ rounds of communication on $G$ (if only clusters of one color class communicate at the same time)  the decomposition is called \emph{$x$-CONGEST-routable}.
\end{definition}
The above definition of a network decomposition relying on weak diameter is suitable for the \LOCAL model where congestion cannot occur: Due to the weak diameter vertices in a cluster can communicate with each other using communication links that are not part of the cluster itself and there cannot be congestion due to different clusters sharing the same edge for communication as message size in the \LOCAL model is unbounded. In the \CONGEST model one also needs to specify the communication structure outside of clusters that vertices use and guarantee that an edge is not used by too many clusters to prevent congestion. As we only use network decomposition in a blackbox manner we do not detail this additional structure; it appears as $x$-\CONGEST-routable in the above definition and is automatically provided by the next theorem.
\begin{theorem}[Network Decomposition of $G^k$, \cite{Portmann19}]
\label{thm:netComp}
There is a deterministic distributed algorithm that in any $N$-node network $G$, which
has $S$-bit identifiers and supports $O(S)$-bit messages for some arbitrary $S$, computes a $(g(N); g(N))$-network decomposition of $G^k$ in $k\cdot g(N) \cdot \log^* S$ rounds, for any $k$ and $g(N) = 2^{O(\sqrt{\log N})}$. The decomposition is $2^{O(\sqrt{\log N})}$-CONGEST-routable. 
\end{theorem}
Second, a \CONGEST algorithm that can list-color graphs efficiently if their diameter, the maximum degree and the color space size are small.\footnote{One can alternatively use the $(deg+1)$-list coloring algorithm of \cite{kuhn20_coloring}  which does not depend on the diameter. However, both algorithm yield the same runtime and the colorspace reduction cannot be avoided in either one.} 

\begin{theorem}[Diameter List Coloring, \cite{BKM19}]
\label{thm:diameterColoring}
There is a deterministic \CONGEST algorithm that given a list-coloring instance $G=(V,E)$ with color space $[C]$, lists $L(v)\subseteq[C]$ for which $|L(v)|\geq\deg(v)+1$ holds for all $v\in V$ and an initial $m$-coloring of $G$, list-colors all nodes in $O\big(D\cdot\log N \log C\cdot (\log\Delta+\log m+\log\log C)\big)$ rounds.

When the result is applied to a subgraph of a communication graph $G$, $N$ refers to the number of nodes in the subgraph, $\deg(v)$ refers to the degree of $v$ in the subgraph and $\Delta$ to the maximum degree of the subgraph, but the diameter $D$ refers to the diameter of $G$. 

The message size of the algorithm is $O(\log C + \log m +\log \Delta)$.
\end{theorem}

We will need additional reasoning to execute the algorithm of \Cref{thm:diameterColoring} on parts of $G^2[U]$ while the communication network is $G$; for that it is essential that we reduce the color space. The core steps of the postshattering phase are as follows.

\subsubsection{Postshattering: Algorithm Overview}
\begin{enumerate}
    \item \textbf{Network decomposition:} Compute a distance-2 network decomposition $\mathcal{D}$ of connected components in graph $K$ using the algorithm of \Cref{thm:netComp} (or an alternative algorithm). 
    \item \textbf{ID space reduction:} Assign new IDs to vertices in $U$ that are unique within each cluster of $\mathcal{D}$. The size of the IDspace is bounded by the cluster size and by $N$.
    \item \textbf{Colorspace reduction:} Within each cluster $\mathcal{C}$ deterministically  compute a colorspace reduction $f_C:[\Delta^2]\rightarrow \poly N$. $f$ is a \emph{colorspace reduction} for the cluster $\mathcal{C}$ if it injectively maps each color list $L_u$ for $u\in \mathcal{C}$.
    
    \textbf{Core Idea:} A random hash function (from a suitable space of hash functions), in expectation, fails for few vertices of the cluster. We derandomize the process of picking such a random hash function with the method of conditional expectation, similar to \cite{CPS17,DKM19, BKM19}.
    \label{step:colorSpacereduction}
    \item \textbf{Final $(deg+1)$-list coloring:} Iterate through the color classes of the network decomposition $\mathcal{D}$ and run the $(deg+1)$-list coloring algorithm of \Cref{thm:diameterColoring} on each cluster.
    Care is needed when refining the lists, i.e., when deleting colors of $d2$-neighbors of previously colored clusters.
    \label{step:listColorComponents}
 \end{enumerate}

\subsubsection{Step 1-2: Network Decomposition, New IDspace}
\label{ssec:NDIDSpace}
We compute a distance-2 network decomposition $\mathcal{D}$ of (the connected components of) $K$ using the algorithm of \Cref{thm:netComp}. \Cref{thm:netComp} can deal with the original $O(\log n)$-bit identifiers and the runtime is upper bounded by $d=2^{O(\sqrt{\log N})}$. We obtain clusters of (weak) diameter $O(d)$ and due to the theorem statement we can simulate one communication round within the cluster, if only clusters with the same color communicate, in $d$ rounds in $G$. Further, the clusters are colored with $d$ colors. 
We assign unique IDs inside each cluster by building a BFS tree inside each cluster (iterating through the color classes of the decomposition $\mathcal{D}$ and handling clusters with the same color in parallel), aggregating the number of nodes of $U$ inside each subtree, and in a convergecast splitting the IDspace $N$ accordingly to the subtrees. The BFS tree might also contain edges and vertices outside of the cluster (see the comment on the additional communication structure after \Cref{def:decomposition}; but only vertices in $U$ are assigned new IDs.
Note that these IDs inside each cluster in particular form a coloring (that we will use for tiebreaking) of the vertices inside a cluster with $N$ colors.

\subsubsection{Step \ref{step:colorSpacereduction}: Color Space Reduction }
\label{ssec:colorSpace}
Throughout this section we fix a cluster $\mathcal{C}$ of the network decomposition. Given, such a cluster we desire to deterministically compute  a \emph{colorspace reduction} $f:[\Delta^2]\rightarrow N^{10}$  of the vertices in the cluster, that is, $f$ is injective on each color list $L_u$ of a vertex $u\in \mathcal{C}$.  
We prove the following result. 
\begin{lemma}[Deterministic Colorspace Reduction]
\label{lem:colorSpaceReduction}
Consider one cluster $\mathcal{C}$ of the network decomposition and let $L_u$ be the list of vertex $u\in \mathcal{C}$ of size $L\leq N$. 
There is a deterministic $2^{O(\sqrt{\log N})}$ round algorithm that computes a colorspace reduction $f: [\Delta^2]\rightarrow [N^{10}]$ such that $|f(L_u)|=|L_u|$ for all $u\in C$.

The colorspace reduction $f$ can be described with $O(\log\log n)$ bits. 
\end{lemma}
To compute a colorspace reduction for all clusters we iterate through the $2^{O(\sqrt{\log N})}$ color classes of the decomposition $\mathcal{D}$ and apply \Cref{lem:colorSpaceReduction} in parallel to all clusters with the same color. The runtime is bounded by $2^{O(\sqrt{\log N})}\cdot 2^{O(\sqrt{\log N})}=2^{O(\sqrt{\log N})}$.

\medskip

Let $N^{10}/2< p < N^{10}$ be a fixed prime which exists due to Bertrand's postulate. We next, define a colorspace reduction $f_e$ for each element  $e\in \mathbb{F}_p$.
Fix $d=N^5$, let $\mathcal{P}_p^d$ be the space of all polynomials over $\mathbb{F}$ of degree $d$ and fix a globally known injective map $\psi: [\Delta^2]\rightarrow \mathcal{P}_p^d$ which exists as $|\mathcal{P}_p^d|=p^{d+1}\geq \Delta^2+1$ ($\psi$ assigns each input color a polynomial). 
Given an element $e\in \mathbb{F}_p$ we define the map $f_e:[\Delta^2]\rightarrow \mathbb{F}_p, x\mapsto (\psi(x))(e)$, that is, color $x$ is first mapped to the polynomial $\psi(x)$ which is then evaluated at position $e$.

We now investigate how 'likely' $f_e$ is a colorspace reduction for a cluster if the element $e\in \mathbb{F}_p$ is chosen uniformly at random. For that purpose, let
 $X_u$ be the random variable that equals $0$ if $|f_e(L_u)|=|L_u|$ and $1$ otherwise.
\begin{lemma}
\label{lem:randReduction}
If $e\in \mathbb{F}_p$ is chosen uniformly at random we have $E\big[\sum_{u\in  \mathcal{C}} X_u\big]\leq 1/N^2$.
\end{lemma}
\begin{proof}
Two distinct colors $x,x'\in [\Delta^2]$ are  mapped to the same element in $\mathbb{F}_p$ if $\psi(x)(e)=\psi(x')(e)$, as $\psi(x)(\cdot)=\psi(x')(\cdot)$ for at most $d$ elements (they are polynomials of degree at most $d< p$ over $\mathbb{F}_p$). Thus the probability for them to map to the same element is upper bounded by $d/p=1/N^5$. With a union bound over all $\binom{|L_u|}{2}$ pairs in the list $L_u$ we obtain
\begin{align}
Pr(X_u=0)\geq 1-\binom{L}{2}(d/p)\geq 1-N^2d/p\geq 1-N^7/N^{10}=1-N^3
\end{align}
Thus we obtain $\E[X_u]=Pr(X_u=1)\leq 1/N^3$ and the claim follows with linearity of expectation and because the cluster size $|\mathcal{C}|$ is bounded by $N$. 
\end{proof}

\begin{proof}[Proof of \Cref{lem:colorSpaceReduction}]
To compute a colorspace reduction for a cluster we use the method of conditional expectation to perform a bitwise derandomization of the following process:
A random bit seed of length $\ell=\lceil\log_2 p\rceil$ is picked uniformly at random. 
If the corresponding value of the seed is  $\geq p$ (interpret the bitstring as an integer represented in base $2$), the process fails, otherwise the process selects an element $e\in \mathbb{F}_p$ and we associate the map $f_e$ with the seed. For technical reasons we introduce a random variable $Y$ which is $0$ if an element is selected, and $1$ otherwise. If an element was selected and $f$ is not a colorspace reduction for $u$  $X_u$ equals $1$, otherwise $X_u$ equals $0$ (in particular in the case that no element was selected).
Let $\Phi=Y+\sum_{u\in \mathcal{C}}{X_u}$. As $\E[Y]\leq 1/2$ and due to \Cref{lem:randReduction} we have  
\begin{align}
    \E[\Phi]=E\big[Y+\sum_{u\in C}X_u\big]\leq 1/2+1/N^2~.
\end{align}
We now use the method of conditional expectation to find a \emph{good seed}, that is, a seed $s$ for which $\E[\Phi\mid \text{ seed fixed to } s]<1$. As there is no randomness involved once the  seed is fixed and the random variables only take  integral values, we obtain $Y=0$ and $X_u=0$ for all $u\in C$. As the process of finding a good seed has moved to the standard repertoire of techniques, e.g., it has been used in \cite{CPS17,BKM19,DKM19,HKM20}, we only sketch it: We iteratively fix the bits of the bitstring beginning with no fixed bit. Assume that bits $1,\ldots,i-1$ are fixed. Setting bit $i$ to $0$ or $1$ splits the remaining probability space into two parts and in one part the expectation of $\Phi$ is at most its expectation on both parts combined. The $i$-th bit is fixed such that its expectation is minimized as follows: Each vertex $u\in \mathcal{C}\cap U$ computes the expectation of $X_u$ in both parts and we aggregate the sum of the expectations at a leader of the cluster; the leader also adds the respective expected values of $Y$. With both sums of expectations at hand the leader can fix the $i$-th bit to the 'better' choice and inform all nodes in the cluster. Due to the properties of the network decomposition the aggregation for one bit of the $\ell=O(\log N)$ bits takes $2^{O(\sqrt{\log N})}$ rounds. Thus, the total runtime is bounded by $\log N\cdot 2^{O(\sqrt{\log N})}=2^{O(\sqrt{\log N})}$.

During this process vertices cannot aggregate the exact expected values of the random variables $X_u$ as the \CONGEST model only allows to send values with a certain precision. However, the precision can be set to be $1/\poly n$ such that the guaranteed expectation of $\Phi$ only increases by an additive $1/\poly n$ in each of the $\ell=O(\log N)$ steps.\footnote{A precision of $1/\poly N$ would be sufficient, but as we have full $O(\log n)$ bits for each message in this aggregation we can also use the better precision of $1/\poly n$.} Thus we are guaranteed that $\E[\Phi \mid \text{seed fixed to }s]$ remains strictly smaller than $1$ and the computed seed is good.
\end{proof}

\subsubsection{Step \ref{step:listColorComponents}: Coloring of Small Components}
We next prove \Cref{lem:postShattering} by iterating through the color classes of the network decomposition and solving the respective $(deg+1)$-list coloring problems, that are obtained by deleting colors from $d2$-neighbors in previously colored clusters from the list. The color space reduction is essential as the runtime of \Cref{thm:diameterColoring} has a $\log C$ factor, which implies a $\Theta(\log \Delta)$ factor in the runtime if we do not reduce the color space.
\begin{proof}[Proof of \Cref{lem:postShattering}]
First, we compute a network decomposition $\mathcal{D}$ of $K$ as detailed in \Cref{ssec:NDIDSpace} and compute a new IDspace of size $2^{O(\log N)}$ for the vertices in $U$ of each cluster; we use this IDspace as an input coloring when applying \Cref{thm:diameterColoring}. First note, that after this IDspace reduction nodes in $U\cap\mathcal{C}$ can in $O(\log \log n)$ rounds learn about their neighbors in the graph $G^2[U\cap \mathcal{C}]$ as each node is equipped with an ID with  $O(\log (2^{O(\sqrt{\log N})}))=O(\log N)$ bits and each node $s\in S$ has at most $O(\log n)$ $U$-neighbors. Thus, for a node $s\in \mathcal{C}$ to inform a neighbor $u$ of $s$ about all its $d2$-neighbors that can be reached through $s$ takes $O(\log\log n \cdot \log n)$ bits in total, which can be sent in $O(\log \log n)$ rounds.

After we have computed the network decomposition $\mathcal{D}$  we perform a colorspace reduction for each cluster as detailed in \Cref{ssec:colorSpace}, in particular, we apply \Cref{lem:colorSpaceReduction} to each cluster. 

To finally list-color the vertices in $U$ we, again, iterate through the $2^{O(\sqrt{\log N})}$ color classes of the network decomposition. In iteration $i$ we color all vertices in clusters of color $i$ and distinct clusters with color $i$ are handled in parallel. 

\textbf{We now detail on iteration $i$ in one cluster $\mathcal{C}$:}
First, vertices refine their list by erasing colors from their list that are used by other $d2$-neighbors that have been colored in iterations $1,\ldots, i-1$: Vertices in $U\cap\mathcal{C}$ broadcast  the cluster's colorspace reduction $f$ for $2$ hops outside the cluster. Each vertex who got colored in iteration $1,\ldots,i-1$ with some color $c$ and receives $f$ computes $f(c)$ and convergecasts this information back to the cluster. Both steps can be implemented with $O(\log \log n)$ overhead as $f$ can be described with $O(\log \log n)$ bits and a vertex has to forward at most $O(\Deltahat)=O(\log n)$ messages of $O(\log\log n)$ size, a total of $O(\log\log n\cdot \log n)$ bits which can be sent in $O(\log\log n)$ rounds. A vertex that receives $f(c)$ removes the value from its list $f(L_u)$. Note, that it might be that a vertex $u$ receives a value $f(c)$ but $c$ was not in $u$'s original list $L_u$; however, removing $f(c)$ from $f(L_u)$ does not hurt as the remaining list size of $u$ remains larger than its uncolored degree (in the cluster). After, each node has refined its list we apply \Cref{thm:diameterColoring} on $G^2[U\cap \mathcal{C}]$. One step of this algorithm can be simulated in $G$ in 
$\poly\log\log n\cdot 2^{O(\sqrt{\log N})}$ rounds, where 
the $2^{O(\sqrt{\log N})}$ term stems from the fact that one round of communication inside the cluster might take $2^{O(\sqrt{\log N})}$ rounds in $G$ due to \Cref{thm:netComp} and the $\poly\log\log n$ term is due to simulating $G^2[U\cap \mathcal{C}]$ in $G[U\cap \mathcal{C}]$, where we use that the message size of the algorithm is bounded by $O(\log\log n)$ bits and an intermediate node has to forward at most $O(\log n)$ messages. The runtime for one cluster can be upper bounded by $\poly\log\log n\cdot 2^{O(\sqrt{\log N})}$ times the following term 
\begin{align}
O\big(D\cdot\log N \log C\cdot (\log\Deltahat+\log m+\log\log C)\big)=O\big(2^{O(\sqrt{\log \log n})}\big)
\end{align}
rounds, where $N=\poly\log n$, $D=2^{O(\sqrt{\log N})}$, $C=\poly N$, $\Deltahat=O(\log n)$ and $m=2^{O(\sqrt{\log N})}$.
Using $N=\poly\log n$  we can bound the total runtime per color class of the network decomposition by $2^{O(\sqrt{\log N})}$ and the total runtime by $2^{O(\sqrt{\log N})}\cdot 2^{O(\sqrt{\log N})}=2^{O(\sqrt{\log \log n})}$.
\end{proof}

\section*{Acknowledgements}
This project was supported by the European Union's Horizon 2020 Research and  Innovation Programme under grant agreement no.~755839 (Yannic Maus) and by the Icelandic Research Fund grant 174484 (Magn\'us M. Halld\'orsson and Alexandre Nolin).
\clearpage
\bibliographystyle{abbrv}
\bibliography{refs}
\appendix

\clearpage
\newcommand{\bud}{\textsc{Buddies}}
\newcommand{\pop}{\textsc{Popular}}
\newcommand{\Hgraph}{H_{2\epsilon}^{\textsc{Pop}}}

\section{Concentration Bounds and Probabilistic Lemmas}
Before we continue with the details of our algorithm we state the following standard Chernoff bound that we utilize frequently in our proofs. 

\begin{proposition}[Chernoff]
Let $X_1, X_2, \ldots, X_n$ be independent Bernoulli trials,
$X = \sum_{i=1}^n X_i$, and $\mu = E[X]$. Then for $\delta>0$,
\begin{align}
\label{eq:chernoff-upper}
\Pr[X \ge (1+\delta)\mu] & 
\le \left(\frac{e^\delta}{(1+\delta)^{1+\delta}}\right)^\mu
\quad\stackrel{(\text{if } \delta\leq 1)}{\le}\quad e^{-\mu \delta^2/3},     \\
\label{eq:chernoff-lower}
\Pr[X \le (1-\delta)\mu] & \le e^{-\mu \delta^2/2}\ . 
\end{align}
\label{P:chernoff}
\end{proposition}
We also use the following inequalities:
\begin{align}
    (1-1/x)^{x-1} & \ge 1/e, \text{ for any } x > 1. \label{eq:inv-e} \\
    1-x & \ge (1/4)^x, \text{ for any } 0 \le x \le 1/2. \label{E:inv-e2}
\end{align}

\section{Implementation of Almost-Clique Decomposition}
\label{app:acd}

We start with the key definitions from \cite{HSS16}, which we adapt to the distance-2 setting\footnote{The authors of \cite{HSS16} used the term $\epsilon$-sparse for what we term $\epsilon$-friendly and $\epsilon$-dense for what we call $\epsilon$-unfriendly.}.
Recall that nodes are $\epsilon$-similar if they have at least $(1-\epsilon) \Delta^2$ d2-neighbors in common.
Observe that if $u$ and $v$ are $\epsilon$-similar and $v$ and $w$ are $\epsilon'$-similar, then $u$ and $w$ are $\epsilon+\epsilon'$-similar.

\begin{definition}
Two nodes are $\epsilon$-\emph{friends} if they are both d2-neighbors and $\epsilon$-similar.
A node is $\epsilon$-\emph{friendly} if it has at least $(1-\epsilon)\Delta^2$ $\epsilon$-friends, and otherwise $\epsilon$-\emph{unfriendly}. 
\end{definition}

We use the following result from \cite{ACK19} that implies that the unfriendly nodes are easily colored by random guesses.
\begin{proposition}[\cite{ACK19}, Prop.~2.2]
An $\epsilon$-unfriendly node is $\epsilon^2 \Delta^2$-sparse.
\label{P:un-sparse}
\end{proposition}

Let $\Vfr_\epsilon$ be the set of $\epsilon$-friendly nodes and let $H^{\text{HSS}}_\epsilon = (\Vfr_\epsilon, E^{\text{HSS}}_\epsilon)$ be a graph on $\Vfr_\epsilon$ with edges between $\epsilon$-friends.
The following lemma captures the key properties of these graphs.
\begin{lemma}[\cite{HSS16}]
Assume $\epsilon \le 1/5$. Let $C$ be a connected component of $H^{\text{HSS}}_\epsilon$.
Each vertex has at most 
$3\epsilon \Delta^2$ non-neighbors in $C$ (i.e., $|C \setminus N_{G^2}(v)| \le 3 \epsilon \Delta^2$).
Furthermore, the nodes in $C$ are mutually $2\epsilon$-similar. 
\label{L:HSS}
\end{lemma}

It is easy to verify the friendship and friendliness properties in the {\LOCAL} model by examining the whole subgraph within distance 4 from a given node. In {\CONGEST}, however, we must be more circumspect. Instead, we determine these properties only approximately.

\begin{definition}
The predicate \bud$_\epsilon(u,v)$ is true if the nodes are $\epsilon$-friends, and false if they are not $2\epsilon$-friends. When neither case applies, the predicate can return either value. If the predicate holds true, we say that the nodes are $\epsilon$-buddies. 

The predicate \pop$_\epsilon(v)$ is true if $v$ has at least $(1-\epsilon)\Delta^2$ $\epsilon$-buddies, and false if it has fewer than $(1-2\epsilon)\Delta^2$ $\epsilon$-buddies. A node $v$ is \emph{$\epsilon$-popular} if \pop$_\epsilon(v)$ holds, and \emph{$\epsilon$-unpopular} otherwise.
\end{definition}

Observe that an $\epsilon$-unpopular node is also $\epsilon$-unfriendly, while an $\epsilon$-popular node is guaranteed to be $2\epsilon$-friendly.

We adapt a method of \cite{ACK19} from the streaming setting to implement the above predicates efficiently in distance-2 {\CONGEST} setting.

\begin{lemma}
There is a $O(\log n)$ round randomized {\CONGEST} algorithm that implements the \bud$_\epsilon$ and \pop$_\epsilon$ predicates, for any fixed $\epsilon > 0$, w.h.p.
Implementing these predicates means that each node knows if it is $\epsilon$-popular, and it knows which pairs $(u,v)$ of its immediate neighbors satisfy \bud$_\epsilon(u,v)$.
\label{L:buddies}
\end{lemma}

\begin{proof}
When $\Delta^2 = O(\log n)$, the nodes can learn of all d2-neighbors and their d2-neighbors in $\Delta^2 = O(\log n)$ rounds. This allows them to compute exactly the friends and friendliness relations.
We focus from now on the case that $\Delta^2 \ge N$, where $N = c_{10}\log n$, for appropriate constant $c_{10}$.

We first implement \bud$_\epsilon$. 
Each node chooses independently with probability $p = c_{10}(\log n)/\Delta^2$ whether to enter a set $S$. Nodes in $S$ inform their d2-neighbors of that fact. For each node $v$, let $S_v$ be the set of d2-neighbors in $S$. W.h.p., $|S_v| = O(\log n)$ (by Prop.~\ref{P:chernoff}). Each node $v$ informs its immediate neighbors of $S_v$, by pipelining in $O(\log n)$ steps.
Note that a node $w$ can now determine the intersection $S_{vu} = S_v \cap S_u$, for its immediate neighbors $v$ and $u$.
Now $w$ determines that $u$ and $v$ are $\epsilon$-buddies iff 
$|S_{vu}| \ge (1-\sqrt{2}\epsilon) N$. 

Let $I_{uv} = G^2[u]\cap G^2[v]$ be the intersection of the d2-neighborhoods of $u$ and $v$. 
For each $w \in I_{uv}$, let $X_w$ be the indicator r.v.\ that $w$ is
selected into the random sample $S$ and let $X = \sum_{w \in I_{uv}} X_w = |S_{uv}|$. Note that $\mu = E[X] = N/\Delta^2 \cdot |I_{uv}|$.

First, suppose $|I_{uv}| \ge (1-\epsilon) \Delta^2$. Then, $\mu \ge (1-\epsilon) N$.
Observe that $(1-\sqrt{2}\epsilon)N \le ((1-2\epsilon)/(1-\sqrt{2}\epsilon))\mu \le (1-(2-\sqrt{2})\epsilon)\mu$,
Then, setting $c_{10} \le 10/((2-\sqrt{2})^2\epsilon^2 (1-\epsilon)$, we have that
the probability that the algorithm incorrectly identifies $u$ and $v$ as $\epsilon$-non-buddies is at most
  \[ \Pr[|S_{vu}| \le (1-\sqrt{2}\epsilon) N] = \Pr[X \le (1-(2-\sqrt{2})\epsilon)\mu] \le e^{-(2-\sqrt{2})^2\epsilon^2/2 \cdot \mu} \le e^{5\ln n} = n^{-5} \ , \]
using (\ref{eq:chernoff-lower}).

Second, let $Q = |I_{uv}|/\Delta^2$ and note that $\mu = Q \cdot N$.
Suppose $|I_{uv}| \le (1-2\epsilon) \Delta^2$, i.e., $Q \le (1-2\epsilon)$.
Let $\delta = (1-\sqrt{2}\epsilon)/Q-1 \ge (1-\sqrt{2}\epsilon)/(1-2\epsilon)-1 \ge (2-\sqrt{2})\epsilon$.
Note that $4/3 \mu \ge (1+\delta) \mu = (1-\sqrt{2}\epsilon)N$.
Then, setting $c_{10} = 20/((1-\sqrt{2}\epsilon)(2-\sqrt{2})^2\epsilon^2)$ and applying (\ref{eq:chernoff-upper}), we have that
the probability that the algorithm incorrectly identifies $u$ and $v$ as $\epsilon$-buddies is at most
  \begin{align*} \Pr[|S_{vu}| \ge (1-\sqrt{2}\epsilon) N] &= \Pr[X \ge (1+\delta)\mu] 
    \le e^{-\delta^2 \mu/3} \\& \le e^{-\delta^2/4 \cdot (1+\delta)\mu} = e^{-\delta^2/4 \cdot (1-\sqrt{2}\epsilon)N} \\
&\le e^{-5\ln n} = n^{-5} \ . \end{align*}

Implementing \pop$_\epsilon$ is nearly identical. The nodes opt into a random sample $T$ with probability $p = N/\Delta^2$. Each node $v$ then gathers $T_v = T \cap N_{G^2}(v)$ from its d2-neighbors, which is of size $O(\log n)$, w.h.p. They distribute their set to their immediate neighbors, who inform them which of the elements in $T_v$ are $\epsilon$-buddies and compute $q_v = |T_v \cap \{u : \bud_\epsilon(u,v)\}|$, the number of $\epsilon$-buddies in the random sample. Node $v$ then determines that it is $\epsilon$-popular iff $q_v \ge (1-\sqrt{2}\epsilon)N$. The correctness is identical to that of \bud$_\epsilon$.
\end{proof}


\noindent\textbf{Lemma \ref{lem:acd}.} \emph{There is a $O(\log n)$ round {\CONGEST} algorithm to form an almost-clique decomposition, for any fixed $\epsilon > 0$. Afterwards, each node knows its component number.}

\begin{proof}
We first implement \bud$_{2\epsilon}$ and \pop$_{2\epsilon}$
to obtain the graph $\Hgraph$ on the $2\epsilon$-popular nodes. We also identify the $\epsilon/2$-popular nodes by implementing \bud$_{\epsilon/2}$ and \pop$_{\epsilon/2}$, 
Let $C_1, C_2, \ldots, C_k$ be the components of $\Hgraph$ that contain an $\epsilon/2$-popular node, and let $V_\ast = V \setminus \bigcup_i C_i$ be the remaining nodes (both those outside $V(H^{\textsc{pop}})$ and those in components of $H^{\textsc{pop}}$ that don't contain an $\epsilon/2$-popular node). Let $\hat{C}_i = C_i \cup \{u \in V_\ast : \exists v \in C_i, \text{\textsc{Buddy}}_{\epsilon/2}(u,v) \}$ be the component extending $C_i$ with all the $\epsilon/2$-buddies (in $V_\ast$) of nodes in $C_i$.
We claim that this yields the desired almost-clique decomposition with parameter $\epsilon$. 

All nodes in $V_\ast$ are $\epsilon/2$-unfriendly, since by definition no node in $V_\ast$ is $\epsilon/2$-popular. This implies Property 1 of Def.~\ref{D:acd}, by Prop.~\ref{P:un-sparse}.
Property 2(e) follows immediately from the construction of $\hat{C_i}$.
We proceed with the rest of Property 2.

Consider a connected component $C_i$ and let $v$ be a $\epsilon/2$-popular node in $C_i$. As remarked earlier, $v$ is then $\epsilon$-friendly.
Let $S_v$ be the set of $2\epsilon$-friends of $v$ and note that 
$|S_v| \ge (1-\epsilon)\Delta^2$ (since it has at least that many $\epsilon$-friends). Any vertex $w$ in $S_v$ has at least $(1-\epsilon)\Delta^2$ common d2-neighbors with $v$, and of those, at most $\epsilon\Delta^2$ are not in $S_v$ (since fewer than $\epsilon \Delta^2$ of $v$'s d2-neighbors are not its $\epsilon$-friends). Thus, $w$ has at least $(1-2\epsilon)\Delta^2$ d2-neighbors in $S_v$. Further, any pair of nodes in $S_v$ is $2\epsilon$-similar, since both are $\epsilon$-friends of $v$.
Hence, each node in $S_v$ is $2\epsilon$-friendly and belongs to $C_i$. It follows that $|C_i| \ge |S_v| \ge (1-2\epsilon)\Delta^2$, establishing Property 2(a).

Observe that $C_i$ is connected in $\Hgraph$ by $2\epsilon$-buddy relationships.
Since $2\epsilon$-buddies are also $4\epsilon$-friends, $C_i$ is also connected in $H_{4\epsilon}^{HSS}$. Thus, by Lemma \ref{L:HSS}, nodes in $C_i$ are mutually $8\epsilon$-similar. Since each node in $\hat{C}_i \setminus C_i$ is $\epsilon$-similar to a node in $C_i$, it follows by transitivity that nodes in $\hat{C}_i$ are mutually $10\epsilon$-similar, establishing Property 2(b).

It also follows from Lemma \ref{L:HSS} that nodes in $C_i$ have at most $12\epsilon\Delta^2$ non-neighbors in $C_i$. So, $C_i$ contains at most $(1+12\epsilon)\Delta^2$ nodes. Each node $u$ in $C_i$ has at least $(1-9\epsilon)\Delta^2$ d2-neighbors in $C_i$, since $u$ is $8\epsilon$-similar to the $\epsilon/2$-popular node $v$ and $v$ has at least $(1-\epsilon)\Delta^2$ d2-neighbors in $C_i$. 
Each node in $\hat{C_i}$ therefore has at least $(1-10\epsilon)\Delta^2$ d2-neighbors in $C_i$, since it is $\epsilon$-similar to a node in $C_i$, establishing Property 2(d). 
Furthermore, there are at most $9\epsilon\Delta^2 |C_i| \le 9\epsilon (1+12\epsilon)\Delta^4$ 2-paths with one endpoint in $C_i$ and the other in $\hat{C_i}$. Hence, 
\[ |\hat{C_i}\setminus C_i| \le \frac{9\epsilon(1+12\epsilon)\Delta^4}{(1-10\epsilon)\Delta^2} 
\le \frac{9\epsilon(1+12\epsilon)\Delta^2}{(1-10\epsilon)}
\le 16\epsilon \Delta^2\ , \]
using that $\epsilon \le 1/40$. Thus, each node in $\hat{C}_i$ has at most $(12+16)\epsilon\Delta^2$ non-neighbors in $\hat{C}_i$, establishing Property 2(c).

We next show that the sets $\hat{C_i}$ are disjoint. 
Let $u$ be a node in $\hat{C}_i \setminus C_i$. It is $\epsilon$-similar to a node $v$ in $C_i$, who is $8\epsilon$-similar to an $\epsilon$-popular node $v'$ in $C_i$. At least $(1-\epsilon)\Delta^2$ d2-neighbors of $v'$ are in $C_i$. So, $u$ and $v'$ have $(1-9\epsilon)\Delta^2$ common d2-neighbors, and at least $(1-10\epsilon)\Delta^2$ of those are in $C_i$. Since $\epsilon < 1/20$, $u$ cannot be $\epsilon$-similar to a node in another component $C_j$. Hence, the $\hat{C_i}$ are disjoint. 

Finally, we need the nodes to learn the component ID in which they belong. First, observe that
any pair of nodes in $\hat{C_i}$ has a common neighbor, since they are $10\epsilon$-similar and $\epsilon < 1/10$. Hence, $\Hgraph[C_i]$ has diameter 2. 
In four rounds, the nodes can then identify as leader the node with the smallest ID of an $\epsilon/2$-popular node and let it define the ID of the component.
\end{proof}

Note that a spanning tree of each component of depth 4 (in $G$) can be formed for aggregation purposes as a BFS tree from the leader of the component.

\section{Small and Large Maximum Degree Case}
\label{app:smallLarge}
Recall that we consider three regimes for $\Delta$: 

\begin{description}
\item[Large degree:] $\Delta \in 2^{\Omega(\log n)}$,
\item[Small degree:] $\log n \cdot \polyloglog n$,
\item[Intermediate degree:] $2^{o(\log n)} \cup \Omegatilde(\log n)$.
\end{description}

In the main text, we have discussed the most interesting intermediate degree case in detail. We now also discuss the much simpler large and small degree cases. The large degree case is the simplest: in this case, we run the $O(\log n)$ algorithm from Theorem~\ref{thm:d2ColoringRand}, described in Section~\ref{sec:randAlg}, as $O(\log n) = O(\log \Delta)$ in this case.

The small degree case is also much simpler than the intermediate degree case, but still requires some work. The high-level idea is that the postshattering phase we describe for the intermediate degree requires a set of conditions that does not include any assumption about the value of $\Delta$, and can be obtained in a much simpler way than in the intermediate degree case with simple compression tricks like sending a many colors in a single \CONGEST message. 

Suppose $\Delta=\log n\cdot \poly\log\log n$. In this section, we prove that there is an $\polyloglog n$-rounds algorithm to get all the preconditions necessary to the execution of the postshattering phase described in~\Cref{sec:postshattering}, thus yielding an algorithm of $2^{O(\sqrt{\log\log n})}$ rounds in total. Said differently, we describe an algorithm that replaces Steps~\ref{step:acd} to~\ref{step:steiner} of the algorithm for the intermediate degree case.

\paragraph*{Degree reduction, graph splitting}

In our algorithm for the intermediate degree case, step~\ref{step:degreereduction} and~\ref{step:degreeestimation} serve to obtain guarantees similar to that of Lemma~5.4 in~\cite{BEPS12}: that we get a graph whose uncolored nodes are split into two sets $U^{lo}$ and $U^{hi}$ such that nodes of $U^{lo}$ have small uncolored degree $O(\log n)$ and that $U^{hi}$ induces a subgraph also of small degree $O(\log n)$. For $\Delta \in \Otilde(\log n)$, we can achieve this with a direct simulation of the Lemma~5.4 in~\cite{BEPS12}.

In this regime of $\Delta$, we have bandwidth to send/receive a $O(\log\log n)$-bit message to and from each d2-neighbor in $\poly\log\log n$ rounds. Thus, the graph splitting process of~\cite{BEPS12} can be efficiently (with multiplicative $\poly\log\log n$ overhead) simulated on $G^2$ using that the original algorithm only needs to send messages of size $\log \Delta \in O(\log\log n)$ in each round, as nodes' messages only consists of colors they are either trying or taking. The nodes can also learn whether their ID is greater or smaller than those of their d2-neighbors' in 2 rounds at the beginning and we have enough bandwidth to keep track of the colors of the neighbors.

Thus, after $O(\log \Delta\cdot \polyloglog n)=O(\poly\log\log n)$ rounds we obtain a partition of the sets of uncolored nodes into two sets $U^{lo}$ and $U^{hi}$, with the desired properties, w.h.p. In addition, all the nodes know their palette exactly, since the algorithm we simulated on $G^2$ has the nodes broadcast their color to all their neighbours upon coloration.

\paragraph*{Proceeding with low and high degree nodes}

As in our algorithm for the intermediate degree case and in~\cite{BEPS12}, we first proceed to color the nodes of $U^{lo}$ before coloring the nodes of $U^{hi}$. The next steps are thus run twice, and in then, $U$ must be understood as being $U^{lo}$ the first time they are run, and $U^{hi}$ the second, and thus to be a set of uncolored nodes of uncolored d2-degree $O(\log n)$ in both runs.

Notice that after coloring $U^{lo}$, the nodes of $U^{lo}$ can broadcast their new color in only $O(\polylog n)$. Thus, when doing all the steps again with $U^{hi}$, we can still assume that they all know their palette perfectly.

\paragraph*{Shattering in small connected components}

In our algorithm for the intermediate degree case, step~\ref{step:shattering} serves to obtain guarantees similar to that of Lemma~5.3 in~\cite{BEPS12}. For $\Delta \in \Otilde(\log n)$, we can simulate Lemma~5.3 in~\cite{BEPS12} in $O(\polylog n)$ rounds directly just like we did with Lemma~5.4 before. Thus, starting from a subset of the uncolored nodes $U$ of maximum uncolored degree $O(\log n)$, applying an $O(\polylog n)$ rounds algorithm we are able to randomly color nodes of $U$ such that the connected components of $G^2[U]$ are of size $O(\polylog n)$ w.h.p.

\paragraph*{Steiner nodes}

We are almost ready to apply the postshattering phase described in~\Cref{sec:postshattering} to $U$, only missing Steiner nodes to make the connected components $G^2[U]$ also connected in $G$, keeping their size $O(\polylog n)$ (this is step~\ref{step:steiner} of our algorithm when in the intermediate degree case). This is easily achieved by having each node of $U$ pick all its direct neighbors as Steiner nodes. This choice of Steiner nodes $S$ makes every connected component in $G^2[U]$ connected in $G[U \cup S]\setminus E(G[S])$, and keeps the size of every connected component in $O(\polylog n)$ since each node of a connected component in $G^2[U]$ adds at most $\Delta \in \Otilde(\log n)$ Steiner nodes to the set of Steiner nodes.

\paragraph*{Postshattering}
Having met all the preconditions necessary to apply the postshattering phase to $U$, we do so as described in~\Cref{sec:postshattering}, in exactly the same way we do when $\Delta$ is in the intermediate regime. This is the only step of the algorithm that takes $\omega(\polylog n)$ rounds.



\end{document}